\documentclass[11pt, english]{article}

\usepackage[margin= 21mm,bottom=21mm, top= 20mm]{geometry}

\usepackage[square,sort,comma,numbers]{natbib}
\usepackage{amsthm}
\usepackage{amsmath}
\usepackage{amssymb}
\usepackage{setspace}
\usepackage{mathtools}
\usepackage{graphicx}
\graphicspath{ {./images/} }
\usepackage[hidelinks]{hyperref}
\usepackage{cleveref}
\usepackage{hyperref}
\usepackage{enumitem}
\usepackage{framed}
\usepackage{subcaption}
\usepackage{xspace}
\usepackage{thmtools} 
\usepackage{thm-restate}

\usepackage{thmtools}
\usepackage{thm-restate}

\usepackage{floatrow}

\usepackage{tikz}
\usepackage{mathdots}
\usepackage{xcolor}
\usepackage{diagbox}
\usepackage{colortbl}
\usepackage[absolute,overlay]{textpos}

\usetikzlibrary{calc}
\usetikzlibrary{decorations.pathreplacing}
\usetikzlibrary{positioning,patterns}
\usetikzlibrary{arrows,shapes,positioning}
\usetikzlibrary{decorations.markings}

\tikzstyle{edge}=[very thick]
\definecolor{bostonuniversityred}{rgb}{0.8, 0.0, 0.0}
\definecolor{arsenic}{rgb}{0.23, 0.27, 0.29}
\tikzstyle{diredge}=[postaction={decorate,decoration={markings,
		mark=at position .56 with {\arrow[scale = 1.3, black]{stealth};}}}]
\tikzset{
    arrow/.style={decoration={markings, mark=at position 0.5 with
    {\fill(-0.18*#1,-0.06*#1) -- (0,0) -- (-0.18*#1,0.06*#1) -- cycle;}, black}, postaction={decorate}},
    arrow/.default=1
}
\tikzset{
    arow/.style={decoration={markings, mark=at position 1 with
    {\fill(-0.09*#1,-0.03*#1) -- (0,0) -- (-0.09*#1,0.03*#1) -- cycle;}}, postaction={decorate}},
    arow/.default=1
}
\tikzset{
    arrrow/.style={decoration={markings, mark=at position 0.9 with
    {\fill(-0.09*#1,-0.03*#1) -- (0,0) -- (-0.09*#1,0.03*#1) -- cycle;}}, postaction={decorate}},
    arow/.default=1
}

\newcommand{\fitellipsis}[2] 
{\draw [fill=white]let \p1=(#1), \p2=(#2), \n1={atan2(\y2-\y1,\x2-\x1)}, \n2={veclen(\y2-\y1,\x2-\x1)}
    in ($ (\p1)!0.5!(\p2) $) ellipse [ x radius=\n2/2+0cm, y radius=1.1cm, rotate=\n1];
}
\newcommand{\Fitellipsis}[2] 
{\draw [fill=white]let \p1=(#1), \p2=(#2), \n1={atan2(\y2-\y1,\x2-\x1)}, \n2={veclen(\y2-\y1,\x2-\x1)}
    in ($ (\p1)!0.5!(\p2) $) ellipse [ x radius=\n2/2+0cm, y radius=1.4cm, rotate=\n1];
}

\theoremstyle{plain}

\newtheorem*{thm*}{Theorem}
\newtheorem{thm}{Theorem}[section]
\Crefname{thm}{Theorem}{Theorems}

\newtheorem*{lem*}{Lemma}
\newtheorem{lem}[thm]{Lemma}
\Crefname{lem}{Lemma}{Lemmas}

\newtheorem*{claim*}{Claim}
\newtheorem{claim}[thm]{Claim}
\crefname{claim}{Claim}{Claims}
\Crefname{claim}{Claim}{Claims}

\newtheorem{prop}[thm]{Proposition}
\Crefname{prop}{Proposition}{Propositions}

\Crefname{remar}{Remark}{Remarks}

\newtheorem{cor}[thm]{Corollary}
\crefname{cor}{Corollary}{Corollaries}

\newtheorem*{conj*}{Conjecture}
\newtheorem{conj}[thm]{Conjecture}
\crefname{conj}{Conjecture}{Conjectures}

\Crefname{qn}{Question}{Questions}

\newtheorem*{obs*}{Observation}
\newtheorem{obs}[thm]{Observation}
\Crefname{obs}{Observation}{Observations}

\Crefname{ex}{Example}{Examples}

\theoremstyle{definition}

\Crefname{prob}{Problem}{Problems}

\newtheorem{defn}[thm]{Definition}
\Crefname{defn}{Definition}{Definitions}

\theoremstyle{remark}

\renewenvironment{proof}[1][]{\begin{trivlist}
\item[\hspace{\labelsep}{\bf\noindent Proof#1.\/}] }{\qed\end{trivlist}}

\newcommand{\remove}[1]{}

\newcommand{\ceil}[1]{
    \left\lceil #1 \right\rceil
}
\newcommand{\floor}[1]{
    \left\lfloor #1 \right\rfloor
}

\newcommand{\F}{\mathbb{F}}
\newcommand{\N}{\mathbb{N}}
\newcommand{\dist}{\mathsf{dist}}
\newcommand{\poly}{\mathsf{poly}}

\title{\vspace{-1 cm}
Small Even Covers, Locally Decodable Codes and Restricted Subgraphs of Edge-Colored Kikuchi Graphs}
\date{}

\author{
Jun-Ting Hsieh\thanks{\texttt{juntingh@cs.cmu.edu}, Carnegie Mellon University}
\and 
Pravesh K. Kothari \thanks{\texttt{kothari@cs.princeton.edu}, Princeton University. Supported by NSF CAREER Award \#2047933, NSF \#2211971, an Alfred P. Sloan Fellowship, and a Google Research Scholar Award. \textbf{}}
\and 
Sidhanth Mohanty\thanks{\texttt{sidhanth@csail.mit.edu}, MIT}
\and
David Munh\'a Correia\thanks{
Department of Mathematics, ETH, Z\"urich, Switzerland. Research supported in part by SNSF grant 200021-228014.
\newline
\emph{Emails}: \texttt{\{david.munhacanascorreia, benjamin.sudakov\}@math.ethz.ch}.
}
\and
Benny Sudakov\footnotemark[4]}
\begin{document} 
\maketitle

\begin{abstract}
Given a $k$-uniform hypergraph $H$ on $n$ vertices, an even cover in $H$ is a collection of hyperedges that touch each vertex an even number of times. Even covers are a generalization of cycles in graphs and are equivalent to linearly dependent subsets of a system of linear equations modulo $2$. As a result, they arise naturally in the context of well-studied questions in coding theory and refuting unsatisfiable $k$-SAT formulas. Analogous to the irregular Moore bound of Alon, Hoory and Linial~\cite{AlonHL02}, Feige conjectured~\cite{Feige08} an extremal trade-off between the number of hyperedges and the length of the smallest even cover in a $k$-uniform hypergraph. This conjecture was recently settled up to a multiplicative logarithmic factor in the number of hyperedges~\cite{GuruswamiKM22,HsiehKM23}. These works introduce the new technique that relates hypergraph even covers to cycles in the associated \emph{Kikuchi} graphs. Their analysis of these Kikuchi graphs, especially for odd $k$, is rather involved and relies on matrix concentration inequalities.

In this work, we give a simple and purely combinatorial argument that recovers the best-known bound for Feige's conjecture for even $k$. We also introduce a novel variant of a Kikuchi graph which together with this argument improves the logarithmic factor in the best-known bounds for odd $k$. As an application of our ideas, we also give a purely combinatorial proof of the improved lower bounds~\cite{AlrabiahGKM23} on 3-query binary linear locally decodable codes. 
\end{abstract}

\section{Introduction}
A set $S$ of hyperedges in a hypergraph $\mathcal{H}$ is an \emph{even cover} if 
$$\bigoplus_{E \in S} E := \{v \in v(\mathcal{H}) : v \text{ belongs to an odd number of hyperedges in $S$}\} = \emptyset .$$
Equivalently, $S$ is an even cover if $\sum_{E \in S} v_E = \mathbf{0}$ over $\F_2$, where $v_E$ denotes the characteristic vector of $E$. In this work, we are interested in understanding the extremal trade-offs between the \emph{size} of a $k$-uniform hypergraph $H$ and \emph{girth}, i.e., the number of hyperedges in the shortest even cover in it. 

\paragraph{Even Covers and Linear Dependencies} Even covers in $k$-uniform hypergraphs correspond to linearly dependent subsets of a system of $k$-sparse (i.e., each equation has exactly $k$ non-zero coefficients) linear equations over $\F_2$. To see why, let us associate a variable $x_v$ for each $v \in H$ and the $|E|$-sparse equation $\sum_{v \in E} x_v = b_E$ for $b_E \in \F_2$ with each edge $E$ in $H$. Then, observe that for any even covers $S$, the left hand sides of the equations corresponding to $E \in S$ add up to $0$ and are thus linearly dependent. Thus, size vs girth trade-offs for $k$-uniform hypergraphs correspond to the largest possible size $\ell$ of the minimum linear dependency in a system of linear equations with $m$ equations in $n$ variables. Such $k$-sparse linear equations arise naturally as the parity check equations for \emph{low-density parity check} error correcting codes. The size vs girth trade-offs for $k$-uniform hypergraphs thus correspond to rate vs distance trade-offs for such codes. 

By the equivalence between even covers and linear dependencies, it is clear that every hypergraph with $m \geq n+1$ hyperedges must have an even cover of length at most $n+1$ and this is clearly tight. The natural question is then: How does the trade-off between $m$ and the maximum possible girth look as $m$ increases beyond $n+1$?

\paragraph{Size vs Girth Trade-offs} When $H$ is a $2$-uniform hypergraph, an even cover is simply an even subgraph, i.e., a subgraph with all vertices of even degree. Such a subgraph is a union of edge-disjoint cycles in $H$. An even cover with smallest number of edges must in fact be a simple cycle and thus, its length must be the girth of the graph $H$. The extremal trade-off between the size of $H$ (i.e., the number of edges in the graph) and its girth was conjectured by Boll\'obas and confirmed by Alon, Hoory and Linial~\cite{AlonHL02} who proved the \emph{irregular Moore bound} --- every graph on $n$ vertices with average degree $d$ has girth at most $2 \lceil \log_{d-1}(n) \rceil$. It is an outstanding open problem whether the constant $2$ in this bound can be further improved (for the best known constant, see \cite{LUW}). 

For $k$-uniform hypergraphs with $k>2$, the size vs girth trade-offs were first studied by Naor and Verstraete~\cite{NaorV08} through applications to  rate vs distance trade-offs for LDPC codes discussed above. They showed that every $H$ with $m \geq n^{k/2} \log^{O(1)}(n)$ hyperedges on $n$ vertices must contain an even cover of length $O(\log n)$. The $\log^{O(1)}(n)$ factor was further improved to a $O(\log \log n)$-factor in a subsequent work of Feige~\cite{Feige08}. For $k=2$, this recovers a coarse version of the irregular Moore bound. For $k>2$, however, there is an interesting regime between the two extreme thresholds of $m=n+1$ (with maximum possible girth of $n+1$) and $m \sim n^{k/2} \log^{O(1)}(n)$ (with maximum possible girth of $O(\log n)$). 

\paragraph{Feige's Conjecture} In 2008, Feige~\cite{Feige08} formulated a  conjecture about this in-between regime that suggests a smooth interpolation between the two extremes noted above. 
\begin{conj}
Fix any $k \in \N$. Then, there exists a sufficiently large $C>0$ such that for sufficiently large $n \in \N$ and every $\ell \in \N$, every $k$-uniform hypergraph $H$ with $m \geq C n (\frac{n}{\ell})^{k/2-1}$ hyperedges has an even cover of length at most $O(\ell \log_2 n)$.
\end{conj}
\noindent The quantitative behavior above can be verified for random hypergraphs (up to a multiplicative factor of $\log(n)$ in $m$). Indeed, Feige's conjecture was based on the hypothesis that random hypergraphs are approximately extremal for the purpose of avoiding short even covers. The motivation for this conjecture comes from the question of showing existence of (and/or efficiently finding) polynomial size \emph{refutation witnesses} --- easily verifiable witnesses of unsatisfiability of --- randomly chosen $k$-SAT formulas parameterized by the number of clauses. Feige's conjecture implies that the result of Feige, Kim and Ofek~\cite{FeigeKO06} --- which showed that random $3$-SAT formulas with $m \geq O(n^{1.4})$ clauses admit a polynomial size refutation witness with high probability --- will extend to the significantly more general setting of \emph{smoothed} $3$-SAT formulas, in addition to simplifying the construction and arguments based on the second moment method in~\cite{FeigeKO06}.

Until recently, not much was known about Feige's conjecture except for the work of Alon and Feige~\cite{AlonF09} that showed a suboptimal version for the case of $k=3$ and that of Feige and Wagner~\cite{feige2016generalized} that built an approach to the problem of even covers by viewing them as an instance of generalized girth problems about hypergraphs. In 2022, Guruswami, Kothari and Manohar~\cite{GuruswamiKM22} proved Feige's conjecture up to an additional loss of $\log^{2k}(n)$ multiplicative factor in $m$ via a spectral argument applied to the \emph{Kikuchi} graph, a graph with an appropriate algebraic structure, built from the given hypergraph. Their argument was simplified and tightened to reduce the loss down to a $O(\log n)$ multiplicative factor in $m$ by Hsieh, Kothari and Mohanty in~\cite{HsiehKM23}. In this work, as we will soon discuss, we give a substantially simpler, purely combinatorial argument that recovers their result and improves the logarithmic factors for hypergraphs of odd uniformity. 

\paragraph{Linear Locally Decodable Codes} A binary error correcting code is a map $C: \{0,1\}^m \rightarrow \{0,1\}^n$, where we view the input as a $m$-bit ``message'' and the output as an $n$-bit codeword. By a slight abuse of notation, we use $C$ to also denote the set of all \emph{codewords}: $\{y \mid \exists x \in \{0,1\}^m, C(x) = y\}$. We say that $C$ is \emph{linear} if, when viewing the input and output as $\F_2^m$ and $\F_2^n$, respectively, $C$ is a $\F_2$-linear map. A code $C$ is called $(q,\delta)$-locally decodable (LDC) if, in addition, it admits a \emph{local decoding} algorithm. Such a local decoding algorithm takes as input any target message bit $i$ for $1 \leq i \leq m$ and a corrupted codeword $y$ such that $\dist(y,C(x)) \leq \delta$ for some $x \in \{0,1\}^m$, where $\dist$ counts the fraction of coordinates that $y$ and $C(x)$ differ. The goal of the algorithm is to access at most $q$ locations in $y$ and output $x_i$ correctly with high probability over the choice of the $q$ locations. In other words, the local decoder can decode any bit of the message by reading at most $q$ locations of the received corrupted codeword. 

Locally decodable codes are intensely investigated in computer science (see the survey~\cite{Yekhanin12} for background and applications) with applications to probabilistically checkable proofs, private information retrieval~\cite{Yekhanin10}, and worst-case to average-case reductions in computational complexity theory. They also have deep connections with additive combinatorics and incidence geometry~\cite{Dvir12}. We are typically concerned with codes that are locally decodable with very few queries, such as $q=2$ or $3$, and the fundamental question is the smallest possible $n = n(m)$ such that there is a $(q,\delta)$-binary LDC $C:\{0,1\}^m \rightarrow \{0,1\}^n$. Classical results have essentially completely resolved the case of $q=2$ and we know that a blocklength of $n \leq 2^{O(m)}$ (for a constant $\delta$) can be achieved by Hadamard codes with a matching lower bound~\cite{GoldreichKST06,KerendisW04}. The case of $q=3$ already presents wide gaps, where until recently, the best known lower bound~\cite{GoldreichKST06,KerendisW04} was $n \geq \tilde{O}(m^2)$, while the best known construction~\cite{Yekhanin08,Efremenko09} gives a $3$-query binary linear code with $n \leq \exp (\exp (O(\sqrt{\log m \log \log m})))$. Recently, using spectral refutations via Kikuchi matrices, Alrabiah et.~al.~\cite{AlrabiahGKM23} improved the quadratic bound above to obtain a lower bound of $n \geq m^3/ \poly \log m$ for 3-query binary, locally decodable codes. 

\subsection{Our Results}
The arguments in the most recent works on Feige's conjecture and locally decodable codes are involved and in particular, require the use of matrix concentration inequalities. Our main contribution is a short, purely self-contained combinatorial argument that recovers their result. This technique might be useful for other similar questions and we illustrate it below, by an application to another well studied problem. Our arguments for Feige's conjecture also utilize Kikuchi graphs introduced in the above prior works. We then introduce a new variant of a Kikuchi graph that, when combined with our combinatorial argument, improves on their results for the case of odd $k$. 

\begin{thm}\label{thm:evencover}
For all $k$, there is a sufficiently large $C$ such that the following holds for all sufficiently large $n$ and $k \leq l \leq n$:
\begin{enumerate}
    \item[(i)] If $k$ is even, then every $k$-uniform $n$-vertex hypergraph with at least $C n \left(n/l \right)^{k/2-1} \cdot \log n$ hyperedges contains an even cover of size $O(l \log n)$.
    \item [(ii)] If $k$ is odd and $l \leq n /\log^2 n$, then every $k$-uniform $n$-vertex hypergraph with at least $C n \left(n/l \right)^{k/2-1} \cdot (\log n)^{\frac{1}{k+1}}$ hyperedges contains an even cover of size $O(l \log n)$.
\end{enumerate}
\end{thm}

\noindent We remark that in the above theorem for odd $k$, we require only a very mild restriction on $l$ that $l \leq n /\log^2 n$, which is purely a technical artefact of our proof. In fact, if $l \geq n/100 \log n$, Feige's conjecture holds trivially from a standard linear algebra argument (see Lemma \ref{lem:linalg}). For $n /\log^2 n \leq l \leq n/100 \log n$, one can use the some arguments as in the proof of the above theorem (namely Section \ref{sec:kikuchi}) to show the same bound as in part (i) for this range.

The ideas we developed to prove the above theorem naturally extend to the setting of \emph{linear} binary locally decodable codes. We can use an altered version of these arguments to show the following result previously obtained by Alrabiah et.~al~\cite{AlrabiahGKM23}. As in the above results, their proof involves spectral arguments on signed adjacency matrices of Kikuchi graphs based on matrix concentration inequalities. The proof we give is simple and purely combinatorial. Although the work in \cite{AlrabiahGKM23} deals also with non-linear codes, we remark that all known constructions of locally decodable codes including the ones discussed in the previous section are linear.

\begin{thm} \label{thm:LDC-non-even-cover}
Let $C:\F_2^m \rightarrow \F_2^n$ be a linear map that gives a $3$-query locally decodable code with distance $\delta>0$. Then, $m \leq K n^{1/3} \log n$ for $K = 10^7/\delta^2$. 
\end{thm}

\paragraph{Finding structures in edge-colored graphs}
We finally remark that the results in this paper, which we discussed above, are proven by reducing both our results  to the one of finding a subgraph, satisfying certain restrictions, in an edge-colored graph. Recall that an edge-colored graph is a graph along with a coloring of the edges so that no vertex has two or more edges of same color incident on it. Finding structures in edge-colored subgraphs is a powerful method in Combinatorics and has been used to study several well-known questions. For example, the famous conjecture of Ringel from 1963 says that the edges of the complete graph $K_{2n+1}$ can be decomposed into copies of any tree on $n$ vertices. It was observed by K\"{o}tzig that this problem can be reduced to showing that a certain edge-coloring of $K_{2n+1}$ contains a rainbow copy of any tree on $n$ vertices. The existence of such rainbow trees was recently established in \cite{MPS}. 

Another famous application is the resolution of the Ryser-Brualdi-Stein conjecture which was open for more than 60 years until it was recently solved by Montgomery \cite{montgomery2023proof}. It states that every Latin square $n \times n$ contains a transversal (i.e., a collection of cells which do not share the same row, column or symbol) of size $n-1$. This can be reduced to finding a rainbow matching missing only one color in any proper edge-coloring of the complete bipartite graph $K_{n,n}$. Finally, a well-studied problem in additive combinatorics studies the additive dimension of sets with small doubling, in any group. This problem had been solved for abelian groups by Sanders \cite{sanders}. For general groups, this question was reduced by Alon et.~al \cite{alon2023essentially}
to a well-known problem (see \cite{rainbowturanproblems}) of finding a rainbow cycle in a properly edge-colored graph with sufficiently many edges, for which Alon et.~al provide an almost tight bound.

\section{Preliminaries}
\subsection{Notation and definitions}
\begin{defn}
All logarithms are taken base 2. Given a graph $G$, we let $d(G)$ denote its average degree. Let $\mathcal{H}$ be a hypergraph. We let $V(\mathcal{H})$ denote its vertex set and $E(\mathcal{H})$ its edge-set. Similarly, we let $v(\mathcal{H})$ denote the size of its vertex set and $e(\mathcal{H})$ denote the number of hyperedges. Given a hypergraph $\mathcal{H}$, a \emph{bucket} is a pair $(\mathcal{E},X)$ where $\mathcal{E}$ is a set of hyperedges in $\mathcal{H}$ which all contain the set of vertices $X \subseteq V(\mathcal{H})$. Given a hypergraph $\mathcal{H}$, an \emph{$(m,t)$-bucket decomposition} is a partition of the hyperedges of $\mathcal{H}$ into buckets $(\mathcal{E}_1,X_1), (\mathcal{E}_2,X_2), \ldots$ such that $ |\mathcal{E}_i| = m$ and $|X_i| = t$.
\end{defn}

\subsection{Standard tools}
\noindent In this section, we collect some lemmas which will be useful throughout the paper. The first is the standard tool used for finding subhypergraphs with large minimum degree.
\begin{lem}\label{lem:mindegree}
Let $\mathcal{H}$ be a hypergraph on $n$ vertices and $nd$ hyperedges. Then, there exists an induced sub-hypergraph $\mathcal{H}' \subseteq \mathcal{H}$ with at least $nd/2$ hyperedges and minimum degree at least $d/2$.
\end{lem}
\begin{proof}
We perform the following process. Start with $\mathcal{H}' := \mathcal{H}$, and while it has a vertex $v$ with degree less than $d/2$ in $\mathcal{H}'$, remove $v$ and its incident hyperedges from $\mathcal{H}'$. Notice that at any point, $\mathcal{H}'$ has at least $e(\mathcal{H}) - nd/2 \geq nd/2$ edges and thus, the process must stop. The final $\mathcal{H}'$ is then a sub-hypergraph with at least $nd/2$ edges and minimum degree at least $d/2$.
\end{proof}
\noindent The next lemma is simply an observation on finding trivial bucket decompositions in hypergraphs. 
\begin{lem}
\label{buckets}
Let $\mathcal{H}$ be a hypergraph on $n$ vertices and $nd$ hyperedges. Then, there exists an induced sub-hypergraph $\mathcal{H}' \subseteq \mathcal{H}$ with at least $nd/2$ hyperedges and a $(d/2,1)$-bucket decomposition.
\end{lem}
\begin{proof}
We do the following process. First, set $\mathcal{H}' := \emptyset$. While $e(\mathcal{H}') < nd/2$, consider the hypergraph $\mathcal{H} \setminus \mathcal{H}'$; it has at least $nd/2$ hyperedges and thus there is a vertex $v$ with degree at least $d/2$; then, take a a collection $\mathcal{E} \subseteq \mathcal{H} \setminus \mathcal{H}'$ of $d/2$ hyperedges containing $v$ and add them to $\mathcal{H}'$. Note that at the end of the process the hypergraph $\mathcal{H}'$ is as desired since the collections $\mathcal{E}$ taken at each step give the $(d/2,1)$-bucket decomposition.
\end{proof}
\noindent The next lemma is a simple application of linear algebra in order to show that hypergraphs with sufficiently many hyperedges contain even covers (with no size restriction).
\begin{lem}\label{lem:linalg}
Let $\mathcal{H}$ be an $n$-vertex hypergraph with at least $n+1$ hyperedges. Then, $\mathcal{H}$ contains an even cover.
\end{lem}
\begin{proof}
For each hyperedge $E$ in the hypergraph consider its characteristic vector $v_E$ in $\mathbb{Z}^n_2$. Note that since we have at least $n+1$ vectors, they must be linearly dependent, thus implying an even cover.
\end{proof}
\noindent The next lemma concerns finding colored cycles in edge-colored graphs. Although a variant of it already appeared in \cite{rainbowturanproblems}, we include a the proof for sake of completeness.
\begin{lem}\label{prop:coloredcycle}
Let $G$ be an $n$-vertex graph and $C$ a set of colors so that each edge in $G$ is assigned a set of $s$ colors in $C$. Suppose that for every vertex $v \in G$ and color $c \in C$, the number of edges incident on $v$ whose assigned set of colors contains $c$ is at most $d(G)/20s \log n$. Then, $G$ contains a closed walk, of size at most $2\log n$, such that some color appears exactly once.
\end{lem}
\begin{proof}
Let $l := \log n$ and $r := d(G)/20s \log n$. For sake of contradiction, suppose $G$ contains no such closed walk of size $2l$. We will double-count the number of rainbow paths in $G$ of size $l$ - by \emph{rainbow} we mean that no color is assigned to more than one edge of the path. In order to give a lower bound on the number of rainbow paths, note first that Lemma \ref{lem:mindegree} implies that $G$ contains a subgraph $G' \subseteq G$ with minimum degree at least $d' = d(G)/4$. We can then take a vertex $v \in G'$ and greedily count the number of rainbow paths in $G'$ of the form $vv_2v_3 \ldots v_{l+1}$. Indeed, we have at least $d'$ options for $v_2$; then, given the assumption on $G$ that every color is incident to a vertex in at most $r$ edges and that $G$ has no rainbow cycle of size at most $2l$, we have at least $d' - sr$ options for $v_3$; with the same reasoning we have at least $d' - 2sr$ options for $v_4$ and so on. Concluding the number of such rainbow paths is at least $d' \cdot (d' - sr) \cdot \ldots \cdot (d' - sr(l-1)) > (0.8d')^l$, since $d' - sr(l-1) > d' - d(G)/20 \geq 0.8d'$. 

On the other hand, since there is no closed walk as described above, we can upper bound the number of rainbow paths as follows. For each pair of vertices $x,y$ (there are at most $n^2$ of them), observe that if there is a rainbow path between them using a set of colors $S\subset C$, then any other rainbow path between $x$ and $y$ must use the same set of colors $S$. Now, the set $S$ has size $sl$ and thus, we can upper bound the number of rainbow paths $xz_1z_2 \ldots z_{l}y$ of size $l$ between $x$ and $y$ as follows. Choosing $z_1$ can be done by choosing a color $c_1 \in S$ and then choosing an edge incident on $v$ which contains the color $c_1$. By assumption, there are at most $slr$ such choices and the choice of the rest of the vertices $z_2, \ldots, z_{l-1}$ can be done in the same way. Concluding, there are at most $(slr)^{l}$ such paths. Therefore, there are at most $n^2 \cdot (slr)^{l} = (4slr)^l = (0.8d')^l$ rainbow paths of size $l$, which is a contradiction given the previous paragraph.
\end{proof}

\subsection{Combinatorial Characterization of Locally Decodable Codes} \label{subsec:LDC-to-even-covers}
In this section, we recall standard results that reduce proving a lower bound on the blocklength of a linear locally decodable code to establishing the existence of a special kind of even cover in a properly edge colored hypergraph. 

We will use the following standard reduction from any locally decodable code to one in the \emph{normal} form (with only constant factor differences in the parameters). We direct the reader the monograph~\cite{Yekhanin12} for a more general result that holds even for non-linear codes. We use $\oplus$ to denote addition modulo 2 (i.e., over $\F_2^n$) in the following. 

\begin{prop}[LDCs in normal form, see Theorem 8.1 in~\cite{DvirNotes}] \label{prop:normal-form}
Let $C:\F_2^m \rightarrow \F_2^n$ be a binary, linear $3$-query locally decodable code with distance $\delta>0$. Then, there is a collection of $3$-uniform hypergraph matchings $H_1, H_2, \ldots, H_m$ of size $\geq \delta n/6$ one for each message bit, such that for every $x \in \F_2^m$, for every $i \in [m]$ and every hyperedge $e \in H_i$, $\oplus_{j \in e} C(x)_j = x_i$. 
\end{prop}

\noindent As an immediate corollary of the normal form of a locally decodable code, we obtain the following useful combinatorial characterization:

\begin{lem} \label{lem:LDC-LB-to-single-odd-even-cover}
Let $C:\F_2^m \rightarrow \F_2^n$ be a binary, linear $3$-query locally decodable code with distance $\delta>0$ and let $H_1, H_2, \ldots, H_m$ be the associated matchings as in Proposition~\ref{prop:normal-form} of size $\geq \delta n$. Then, for any even cover $\{e_1, e_2, \ldots, e_t\}$ in the multi-hypergraph $\cup_{i \leq [m]} H_i$, each $H_i$ contributes an even number of hyperedges. 
\end{lem} 
\begin{proof}
Suppose not and say $\{e_1, e_2, \ldots, e_t\}$ is an even cover that contains an odd number of hyperedges from, say, $H_1$. Then, choose $x = (1,0,\ldots,0) \in \F_2^m$ and observe that on the one hand, $\oplus_{i \leq t} \oplus_{j \in e_i} C(x)_j = 0$ because $C_i$s form an even cover. But, on the other hand, $\oplus_{i \leq t} \oplus_{j \in e_i} C(x)_j = C(x)_1 = 1$ since $H_1$ contributes an odd number of hyperedges to $\{e_1, e_2, \ldots, e_t\}$ and for every hyperedge $e \in H_1$, $\oplus_{j \in e} C(x)_j = x_1$. This is a contradiction. 
\end{proof}

\noindent By thinking of each hyperedge in $H_i$ as having the color $i$,  the hypergraph $\cup_{i \leq m} H_i$ is a properly edge colored hypergraph. Thus, in light of Lemma~\ref{lem:LDC-LB-to-single-odd-even-cover}, Theorem~\ref{thm:LDC-non-even-cover} is implied by the following theorem that we will establish in this work. 

\begin{thm}\label{thm:LDC}
For all sufficiently small $\alpha > 0$, there is $K := 10^7/\alpha^2 >0$ such that the following holds for all sufficiently large $n$. Let $\mathcal{H}$ be a $3$-uniform $n$-vertex hypergraph which is properly edge-colored with $Kn^{1/3} \log n$ colors so that each color has at least $\alpha n$ hyperedges. Then, $\mathcal{H}$ contains an even cover such that some color appears exactly once.
\end{thm}

\section{Even \texorpdfstring{$k$}{k}, and some ideas for odd \texorpdfstring{$k$}{k}}
\subsection{Feige's problem for even \texorpdfstring{$k$}{k}}\label{sec:evenk}
To illustrate our main ideas, we start this section by giving a very short solution of Feige's problem up to a logarithmic factor for all even $k$. In the next section we will discuss how to use the methods inspired by this proof for the odd values of $k$. 

\begin{proof}[ of Theorem \ref{thm:evencover} (i)]
Let $\mathcal{H}$ be a hypergraph satisfying the assertion of the theorem. By choosing constant $C=10^k$ we can assume that $\mathcal{H}$ has at least
$10^k n \left(n/l \right)^{k/2-1}  \log n$ edges. We define an edge-colored graph $G$, called the \emph{Kikuchi graph}, as follows.
\begin{itemize}
    \item The vertex set of $G$ consists of all $l$-element subsets of the vertex of $\mathcal{H}$, i.e.,  $V(G) := \binom{[n]}{l}$.
    \item We define an edge $S \xleftrightarrow{} T$ for two $S,T \in V $ if there exists an edge $E \in \mathcal{H}$ such that $S \oplus T = E$ and $|S \cap E| = |T \cap E| = k/2$. Moreover, we color the edge $S \xleftrightarrow{} T$ in $G$ with color $E$.
\end{itemize}

\noindent 
\begin{figure}[ht]
\begin{tikzpicture}[scale = 3]

\node[scale=1, color=black] at (-1.7,0.3){$1$};
\node[scale=1, color=black] at (-1.7,0.7){$2$};    \node[scale=1, color=black] at (-1.3,0.7){$3$};    \node[scale=1, color=black] at (-1.3,0.3){$4$};
\node[scale=1.5, color=black] at (-1.5,0.5){$E$};

\draw [rounded corners=10mm] (-1.9,0.1)--(-1.9,0.9)--(-1.1,0.9)--(-1.1,0.1) -- cycle;

    \node[scale=1, color=black] at (0.4,-0.1){$1$};
    \node[scale=1, color=black] at (0,1.1){$3$};
    \node[scale=1, color=black] at (0.4,1.1){$2$};
    \node[scale=1, color=black] at (0,-0.1){$4$};

    \draw[rounded corners] (-0.1,1.2) rectangle (0.5,-0);
    \draw[rounded corners] (-0.1,1) rectangle (0.5,-0.2);
    \node[scale=1.5, color=black] at (0.2,0.5){$S \cap T$};
    \node[scale=1.5, color=black] at (0.2,1.4){$S$};
    \node[scale=1.5, color=black] at (0.2,-0.4){$T$};

\end{tikzpicture}
\caption{
    An illustration of the Kikuchi graph defined above with an edge $E = \{1,2,3,4\}$ of $\mathcal
    {H}$ and an edge $S \xleftrightarrow{} T$ of $G$ colored with $E$.}
    
\end{figure}
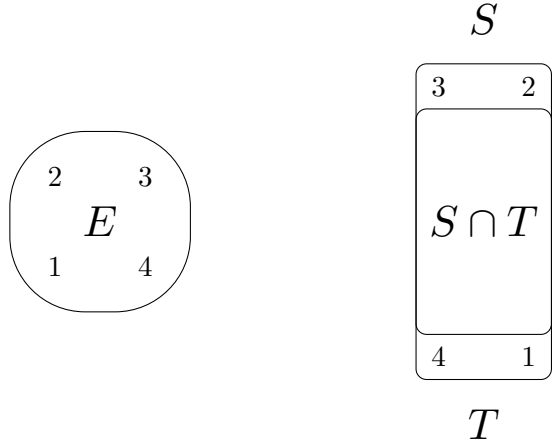

\noindent Let us make some important remarks about $G$. First, note that the coloring of its edges is \emph{well-defined} and \emph{proper}. Indeed, the first is the case since the color of the edge $S \xleftrightarrow{} T$ is uniquely defined by $S \oplus T$. The latter also holds since given $S \in V$ and $E \in \mathcal{H}$, there exists at most one $T \in V$ such that $S \oplus T = E$. We now show the following claim.
\begin{claim}
$G$ contains a closed walk of length $O(l \log n)$ such that some color appears only once.
\end{claim}
\begin{proof}
To prove this we apply Lemma \ref{prop:coloredcycle}. Note that $G$ has $N := {n \choose l}$ vertices and each edge $E \in \mathcal{H}$ creates at least $\frac{1}{2} \cdot \binom{k}{k/2} \cdot \binom{n-k}{l-k/2}$ edges of $G$, since we need to choose the intersections $S \cap E, T \cap E$ (which are disjoint and of size $k/2$) and then the set $S \setminus E=T \setminus E$. Therefore, the number of edges in $G$ is 
$e(G) \geq \frac{1}{2} \cdot e(\mathcal{H}) \binom{k}{k/2} \binom{n-k}{l-k/2}$. Since $N = \binom{n}{l}$ we have $\log N \leq \log n^l=l \log n$ and 
\begin{align*}
N &= \binom{n}{l}
= \binom{n-k}{l-k/2} \cdot \frac{(n-k+1)\cdots n}{(l-k/2+1) \cdots l \cdot (n-l) \cdots (n-l-k/2+1)}\\
&\leq \binom{n-k}{l-k/2} \cdot (4n/l)^{k/2},
\end{align*}
where we used that $l \leq n/ \log n$ by Lemma \ref{lem:linalg} and that $l \geq k$. 
Therefore 
\begin{eqnarray*}
e(G)& \geq & \frac{1}{2} \cdot e(\mathcal{H}) \binom{k}{k/2} \binom{n-k}{l-k/2}\\
&\geq & \frac{1}{2} \cdot 10^k n \left(n/l \right)^{k/2-1} \log n \cdot \binom{k}{k/2} \cdot \binom{n-k}{l-k/2}\\ 
&\geq&   5^k  (4n/l)^{k/2} \binom{n-k}{l-k/2} \big(l \log n\big) >10N\log N.
\end{eqnarray*}
Therefore, since $G$ is properly-colored, Lemma \ref{prop:coloredcycle} implies that it contains the desired type of closed walk of size at most $2\log N \leq 2l \log n$.
\end{proof}

To finish, we can take the closed walk $S_1,S_2, \ldots, S_r$ in $G$ given above and consider the collection $\mathcal{C}$ of edges in $\mathcal{H}$ which appear an odd number of times as colors in this walk - that is, recalling the definition of $G$, those edges $E \in \mathcal{H}$ such that $E = S_i \oplus S_{i+1}$ (the indices being taken modulo $r$) for an odd number of $1 \leq i \leq r$. Since $\bigoplus_{1 \leq i \leq r} (S_i \oplus S_{i+1}) = \emptyset$, we have that $\bigoplus_{E \in \mathcal{C}} E = \emptyset$ and so $\mathcal{C}$ is an even cover of size at most $r = O(l \log n)$, as desired.
\end{proof}
\subsection{Feige's problem for odd \texorpdfstring{$k$}{k} and LDCs}\label{sec:oddkoutline}
\noindent We now briefly discuss the methods for proving Theorem \ref{thm:evencover} in full generality. In the previous section, in order to solve Feige's problem for even $k$, we defined an edge-colored graph $G$ based on the hypergraph $\mathcal{H}$. Crucially, the graph $G$ had the following properties.
\begin{enumerate}
    \item $G$ has many edges and is properly edge-colored.
    \item A closed walk in $G$ in which some color appears once implies an even cover in $\mathcal{H}$.
\end{enumerate}
\noindent Then we just applied Lemma \ref{prop:coloredcycle} to find a the desired even cover in $\mathcal{H}$. 

\vspace{0.2cm}
\noindent \textbf{Even covers with odd $k$}
\vspace{0.2cm}

\noindent In the case of odd $k$ our general strategy will be similar. We will also want to define an edge-colored graph $G$ based on our hypergraph $\mathcal{H}$ such that some appropriate versions of the properties above hold. For this, we will define two different variants of the Kikuchi graph defined in the previous section. The first version is given in Section \ref{sec:kikuchi} and has already appeared in previous works (see \cite{GuruswamiKM22} and \cite{HsiehKM23}) and can alone, be used to show a bound as in part (i) of Theorem \ref{thm:evencover} for odd values of $k$. To get the improved part (ii) of Theorem \ref{thm:evencover} for odd $k$, another variant is necessary. In Section \ref{sec:flowerkikuchi}, we introduce then the \emph{flower Kikuchi graph}. 

Before using these Kikuchi graphs however, we will need to \emph{clean} our hypergraph $\mathcal{H}$ using the tools in Section \ref{sec:cleaning}. Specifically, we will be able to reduce our problem to considering a hypergraph $\mathcal{H}$ with a nice bucket decomposition and co-degree assumptions (see Lemma \ref{prop:cleaning1}). Then, based on the outcome of this Lemma, we will either use the Kikuchi graph on $\mathcal{H}$ defined in Section \ref{sec:kikuchi} or the new flower Kikuchi graph.

In the first case, the graph is edge-colored, but each edge receives now a pair of colors $(C,C')$, where each color $C$ is an edge in $\mathcal{H}$. Due to the additional assumptions we now have on the hypergraph $\mathcal{H}$ we can show that the Kikuchi graph contains a subgraph $G' \subseteq G$ with the following properties similar to the ones before.
\begin{enumerate}
    \item $G'$ has many edges and is such that every vertex is incident to at most $d(G')/80 \log |G'|$ edges containing a given color.
    \item A closed walk in $G'$ in which some color appears once implies an even cover in $\mathcal{H}$.
\end{enumerate}
Then, like before, applying Lemma \ref{prop:coloredcycle} implies the desired even cover in $\mathcal{H}$. As mentioned above, finding $G'$ is only possible because of the additional assumptions on $\mathcal{H}$. In Proposition \ref{prop:kikuchiclean} we state the crucial hypergraph property which implies the existence of such a $G'$.

In the second case, the flower Kikuchi graph $G$ is edge-colored, with each edge receiving a color $C$ which is an edge of $\mathcal{H}$. Like in the previous case, we can also then perform a (simpler) cleaning procedure to find a subgraph subgraph $G' \subseteq G$ with the same properties as those mentioned in the beginning of the section: $G'$ has many edges and is properly colored; a closed walk in $G'$ in which some color appears once implies an even cover in $\mathcal{H}$. Then, applying Lemma \ref{prop:coloredcycle} implies the desired even cover in $\mathcal{H}$.

\vspace{0.2cm}
\noindent \textbf{Locally decodable codes}
\vspace{0.2cm}

\noindent The problem of LDCs is very similar. As discussed in the introduction, we can reformulate this problem as finding a special even cover in an edge-colored hypergraph (see Theorem \ref{thm:LDC}). The only difference here is that the hyperedges of the hypergraph are properly colored and we want to find an even cover where some color appears exactly once. We again will first clean our hypergraph using Lemma \ref{lem:basiccleaning}. Then we will show that the given Kikuchi graph defined in Section \ref{sec:kikuchi} will satisfy the property in Proposition \ref{prop:kikuchiclean} and thus gives the desired even cover in $\mathcal{H}$.
\section{Cleaning tools, the Kikuchi graph and even covers}\label{sec:cleaning}
In this section we present two tools for cleaning hypergraphs. As we mentioned in the previous section, the goal is to find hypergraphs with \emph{nice} bucket decompositions and co-degree conditions. Here and later in the paper we will always assume that $k$ is odd.
\subsection{Cleaning I: A general tool for hypergraphs}
First, we need the following simple observation for general hypergraphs.
\begin{lem}\label{lem:basiccleaning}
Let $\mathcal{H}$ be a $k$-uniform hypergraph and $j \leq k$. Then, for all $t$ and $e \leq e(\mathcal{H})$, there is a $\mathcal{H}' \subseteq \mathcal{H}$ such that one of the following holds.
\begin{enumerate}
\item $e(\mathcal{H}') \geq e(\mathcal{H}) - e$ and every set $R$ of size $t$ has $\text{deg}_{\mathcal{H}'}(R) < m$.
    \item $\mathcal{H}'$ has a $(m, t)$-bucket decomposition and $e(\mathcal{H}') \geq e $.
    
\end{enumerate}
\end{lem}
\begin{proof}
Starting with $\mathcal{H}_0 := \mathcal{H}$, perform the following procedure. While $\mathcal{H}_0$ has a set $R$ of size $t$ with $\text{deg}_{\mathcal{H}_0}(R) \geq m$, then select $m$ such hyperedges in $\mathcal{H}_0$ which contain $R$ and remove them from $\mathcal{H}_0$. When this process stops, if $e(\mathcal{H}_0) \geq e(\mathcal{H}) - e$, then by taking $\mathcal{H}'=\mathcal{H}_0$ we have the first case.
Otherwise by taking $\mathcal{H'}=\mathcal{H}-\mathcal{H}_0$ we have the second case.
\end{proof}
\subsection{Cleaning II: Hypergraphs with high co-degrees have small even covers}
Next, we will prove a particular statement concerning Feige's problem. This will later allow us to reduce the general problem to considering a hypergraph with some co-degree assumptions. 
\begin{lem}
Let $\mathcal{H}$ be an $n$-vertex $k$-uniform hypergraph which has a $\left(2, \frac{k+1}{2} \right)$-bucket decomposition. Then, for sufficiently large cosntant $C$, if $e(\mathcal{H}) \geq C\cdot n \left(n/l \right)^{\frac{k-3}{2}} \cdot \log n$, it contains an even cover of size $O(l \log n)$. 
\end{lem}
\begin{proof}
Let us denote the $\left(2, \frac{k+1}{2} \right)$-bucket decomposition by $(\mathcal{E}_1, X_1), (\mathcal{E}_2, X_2), \ldots$ and denote $\mathcal{E}_i$ as $\{Y_i,Z_i\}$. By the pigeonhole principle we can find some $j \geq \frac{k+1}{2}$ such that considering only the buckets with $|Y_i \cap Z_i| = j$, will reduce the size of the hypergraph by at most a $2/k$ factor. Now, we construct a new $2(k-j)$-uniform multihypergraph $\mathcal{G}$ on the same vertex set as $\mathcal{H}$ consisting of the hyperedges $Y_i \oplus Z_i$. The key observation which is easy to check from this definition is the following.
\begin{obs}
An even cover in $\mathcal{G}$ implies an even cover of twice the same size in $\mathcal{H}$.
\end{obs}
\noindent The above implies that we need only to show that $\mathcal{G}$ contains an even cover of size $O(l \log n)$. For this, note that $2(k-j)$ is even and $e(\mathcal{G}) \geq e(\mathcal{H})/k \geq (C/k) \cdot n \left(n/l \right)^{\frac{k-3}{2}} \cdot \log n \geq(C/k) \cdot n \left(n/l \right)^{(k-j)-1} \cdot \log n$. Thus we can directly imply Theorem \ref{thm:evencover} (i).
\end{proof}

\noindent The above discussion together with the proof of 
 Lemma \ref{lem:basiccleaning} (applied with $j = \frac{k+1}{2}$) immediately imply the following corollary.
\begin{cor}\label{cor}
For sufficiently large constant $C$, every $n$-vertex $k$-uniform hypergraph $\mathcal{H}$ without an even cover of size $O(l \log n)$ has a subhypergraph $\mathcal{H}' \subseteq \mathcal{H}$ such that the following hold.
\begin{itemize}
    \item $e(\mathcal{H}') \geq e(\mathcal{H}) - C \cdot n \left(n/l \right)^{\frac{k-3}{2}} \cdot \log n$.
    \item Every set $S$ of size at least $\frac{k+1}{2}$ has $\text{deg}_{\mathcal{H}'}(S) \leq 1$.
\end{itemize}
\end{cor}
\subsection{Cleaning III: A general tool for hypergraphs}
In this section we present a tool for finding a subhypergraph with both a nice bucket decomposition and co-degree assumptions. Given the lemmas in the previous section, we will need only to apply it to $k$-uniform hypergraphs with odd $k$ in which no two hyperedges share more than $k/2$ vertices.
\begin{lem}\label{prop:cleaning1}
Let $\mathcal{H}$ be an $n$-vertex $k$-uniform hypergraph with $nd$ hyperedges such that no two hyperedges share more than $k/2$ vertices.
Consider any function $m : [\floor{k/2}] \rightarrow \N^+$ with $m(1) \leq d/4k$. Then, there exists a subhypergraph $\mathcal{H}_0 \subseteq \mathcal{H}$ with at least $nd/k$ hyperedges and an index $1\leq t \leq k/2$ such that the following hold.
\begin{enumerate}
    \item[1)] For all $S \subseteq V(\mathcal{H}_0)$ with $t < |S| \leq k/2$, we have that $\text{deg}_{\mathcal{H}_0}(S) < m(|S|)$.
    \item[2)] If $t = 1$, then $\mathcal{H}_0$ has minimum degree at least $m(1)$. If $t \geq 2$, then $\mathcal{H}_0$ has a $\left(m(t),t \right)$-bucket decomposition. 
\end{enumerate}
\end{lem}  
\begin{proof}
Set $\mathcal{H}_2, \ldots, \mathcal{H}_{\floor{k/2}} := \emptyset$ and $\mathcal{H}' := \mathcal{H}$ and repeat the following procedure as long as possible.
Let $r \in [2, \floor{k/2}]$ be maximal such that there exists some $S \subseteq V(\mathcal{H}')$ with $|S| = r$ with $\text{deg}_{\mathcal{H}'}(S) \geq m(r)$. Take 
exactly $m(r)$ hyperedges in $\mathcal{H}'$ that contain $S$, remove them from $\mathcal{H}'$ and add them to $\mathcal{H}_{r}$. The process stops when such an $r$ does not exist.

Now, suppose first that this process goes until $e(\mathcal{H}') < e(\mathcal{H})/2$. Then, we can stop the process at the first step in which that occurs. Note then that by pigenholing, one of the hypergraphs $\mathcal{H}_r$ will have at least $nd/k$ hyperedges. We then define $\mathcal{H}_0 := \mathcal{H}_r$ and observe that it satisfies the desired conditions.

Otherwise, the process stops while $e(\mathcal{H}') \geq e(\mathcal{H})/2 = nd/2$. Note that since the process stopped, it must be that such an $r$ no longer exists and thus, $\text{deg}_{\mathcal{H}'}(S) < m(|S|)$ for all sets $S$ with $|S| \geq 2$. Now, by Lemma \ref{lem:mindegree}, $\mathcal{H}'$ has a subgraph, which we take to be $\mathcal{H}_0$, with at least $nd/4$ hyperedges and minimum degree at least $d/4 \geq m(1)$. Clearly, the conditions now are satisfied for $t = 1$.
\end{proof}

\subsection{The Kikuchi graph and even covers}\label{sec:kikuchi}
We will now introduce the variant of the \emph{Kikuchi graph} for an edge-colored hypergraph with a given $(m,t)$-bucket decomposition which we already mentioned in Section 3.2. 
\begin{defn}\label{def:kikuchi}
Let $\mathcal{H}$ be an $n$-vertex edge-colored $k$-uniform hypergraph with an $(m,t)$-bucket decomposition with buckets $\{(\mathcal{E}_i, X_i)\}_{1 \leq i \leq p}$. 
Given an $l$, we define the \emph{$l$-Kikuchi graph} of $\mathcal{H}$ and this bucket decomposition to be an edge-colored graph $G = (V,E)$ as follows. Let $[n] \times [2]$ denote two disjoint copies of $[n]$, whose vertices are labeled $1$ and $2$ in order to distinguish between these sets. The vertex set $V(G)$ consists of all subsets of $[n] \times [2]$ of size $l$. Each such set $S \in V$ is viewed as $(S^{(1)},S^{(2)})$, where $S^{(1)}, S^{(2)} \subseteq [n]$ have labels $1$ and $2$ respectively. For each $i\in [p]$ and each ordered pair $(C,C')$ of hyperedges $C,C' \in \mathcal{E}_i$, let $\tilde{C}^{(1)}$ be $\tilde{C}:= C \setminus X_i$ labeled $1$ and $\tilde{C}'^{(2)}$ be $\tilde{C'} := C' \setminus X_i$ labeled $2$. We add an edge between $S, T\in V$, denoted $S \xleftrightarrow{(C,C')} T$, if $S \oplus T = \tilde{C}^{(1)} \oplus \tilde{C}'^{(2)}$ and one of the following holds.
    \begin{itemize}
        \item $|\tilde{C}^{(1)} \cap S^{(1)}| = |\tilde{C}'^{(2)} \cap T^{(2)}| = \ceil{\frac{k-t}{2}}$ and $|\tilde{C}'^{(2)} \cap S^{(2)}| = |\tilde{C}^{(1)} \cap T^{(1)}| = \floor{\frac{k-t}{2}}$.

        \item $|\tilde{C}^{(1)} \cap S^{(1)}| = |\tilde{C}'^{(2)} \cap T^{(2)}| = \floor{\frac{k-t}{2}}$ and $|\tilde{C}'^{(2)} \cap S^{(2)}| = |\tilde{C}^{(1)} \cap T^{(1)}| = \ceil{\frac{k-t}{2}}$.
    \end{itemize}
\noindent Further, the edge $S \leftrightarrow T$ is said to be \emph{associated} with the hyperedges $C,C'$ in $\mathcal{H}$ and is colored by the pair $(c,c')$ where $c$ is the color of $C$ in $\mathcal{H}$ and $c'$ is the color of $C'$. 
\end{defn} 

\noindent
Note that the Kikuchi graph defined here is more complicated than the one introduced in Section \ref{sec:evenk}. We split the vertex set $V$ into two copies of $[n]$. This is needed because the pairs $C ,C'$ in our definition of buckets may have intersection larger than $t$, meaning that $|C \oplus C'| = |\tilde{C} \oplus \tilde{C}'| < 2(k - t)$. Taking $\tilde{C}, \tilde{C}'$ to be subsets of two different copies of $[n]$ automatically makes $\tilde{C}^{(1)}, \tilde{C}'^{(2)}$ disjoint, so that $|S \oplus T| = |\tilde{C}^{(1)} \oplus \tilde{C}'^{(2)}| = 2(k - t)$. Note also that we allow a vertex of the Kikuchi graph $S \subseteq [n] \times [2]$ to contain two copies of some element in $[n]$ with different labels.

The conditions in Definition \ref{def:kikuchi} are such that $S^{(1)}$ and $T^{(1)}$ (the vertices in $S$, $T$ labeled $1$) have symmetric difference $\tilde{C}^{(1)}$, while $S^{(2)}$ and $T^{(2)}$ (the vertices in $S$, $T$ labeled $2$) have symmetric difference $\tilde{C}'^{(2)}$.
Moreover, we would like $\tilde{C}^{(1)}$ and $\tilde{C}'^{(2)}$ to be split evenly, but if $k-t$ is odd, then we can only split them into $\ceil{\frac{k-t}{2}}$ and $\floor{\frac{k-t}{2}}$ sized sets, and we must ensure that $S$ and $T$ have the same sizes.

Given the above Kikuchi graph, we also define  associated hypergraphs $\mathcal{H}_c$ which will be crucial for our proof. In Proposition \ref{prop:kikuchiclean} we will show that if these hypergraphs have hyperedges that are rather uniformly distributed, that is, that there is no set of vertices with too large co-degree, then $\mathcal{H}$ must contain an even cover of small size. 

\begin{defn}\label{def:HCdef}
For each color $c$ of $\mathcal{H}$, we define the $(k-t)$-uniform multihypergraph $\mathcal{H}_c$ with vertex set $[n] \times [2]$ as follows: for each $c$-colored hyperedge $C$ in $\mathcal{H}$ belonging to some bucket $\mathcal{E}_i$ and a hyperedge $C' \neq C$ also in $\mathcal{E}_i$, we put in $\mathcal{H}_c$ all hyperedges $F$ such that one of the following holds.
 \begin{itemize}
        \item $\{|\tilde{C}^{(1)} \cap F^{(1)}|, |\tilde{C}'^{(2)} \cap F^{(2)}|\} = \{\ceil{\frac{k-t}{2}}, \floor{\frac{k-t}{2}}\}$.
        \item $\{|\tilde{C}^{(2)} \cap F^{(2)}|, |\tilde{C}'^{(1)} \cap F^{(1)}|\} = \{\ceil{\frac{k-t}{2}}, \floor{\frac{k-t}{2}}\}$.
\end{itemize}
\end{defn}

\noindent
Note that, if $S \xleftrightarrow{(C,C')} T$ is an edge of the Kikuchi graph $G$ and $C, C'$ have colors $c,c'$ respectively, then both $S$ and $T$ contain in addition to common part $S \cap T$ a hyperedge from
$\mathcal{H}_c \cap \mathcal{H}_{c'}$.

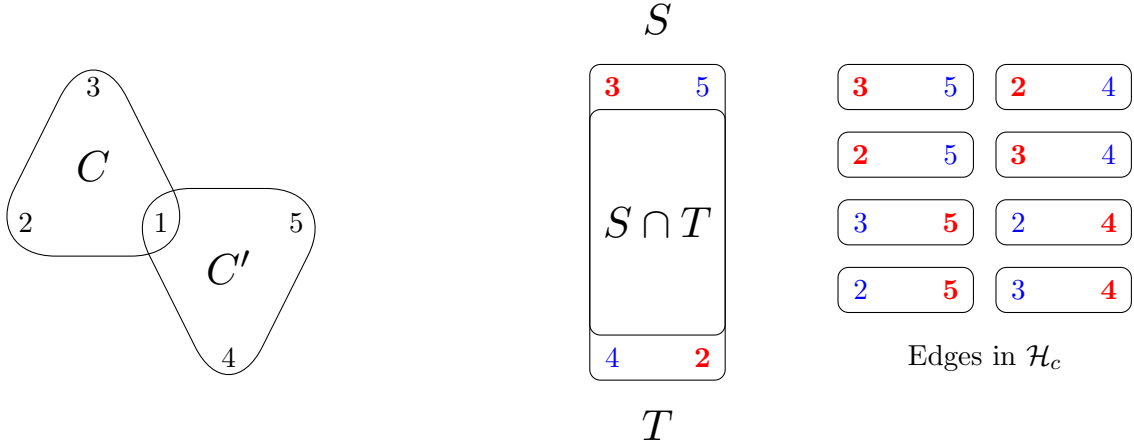
\begin{figure}[ht]
\begin{tikzpicture}[scale = 3]

    \node[scale=1, color=black] at (-2.6,0.5){$2$};
    \node[scale=1, color=black] at (-2,0.5){$1$};
    \node[scale=1, color=black] at (-2.3,1.1){$3$};
    \node[scale=1, color=black] at (-1.4,0.5){$5$};
    \node[scale=1, color=black] at (-1.7,-0.1){$4$};
    \draw [rounded corners=10mm] (-1.7,-0.35)--(-2.2,0.65)--(-1.2,0.65) -- cycle;
    \node[scale=1.5, color=black] at (-1.7,0.3){$C'$};
    \draw [rounded corners=10mm] (-1.8,0.35)--(-2.3,1.35)--(-2.8,0.35) -- cycle;
    \node[scale=1.5, color=black] at (-2.3,0.75){$C$};

    \node[scale=1, color=red] at (0.4,-0.1){$\mathbf{2}$};
    \node[scale=1, color=red] at (0,1.1){$\mathbf{3}$};
    \node[scale=1, color=blue] at (0.4,1.1){$5$};
    \node[scale=1, color=blue] at (0,-0.1){$4$};

    \draw[rounded corners] (-0.1,1.2) rectangle (0.5,-0);
    \draw[rounded corners] (-0.1,1) rectangle (0.5,-0.2);
    \node[scale=1.5, color=black] at (0.2,0.5){$S \cap T$};
    \node[scale=1.5, color=black] at (0.2,1.4){$S$};
    \node[scale=1.5, color=black] at (0.2,-0.4){$T$};

\draw[rounded corners] (1,1.2) rectangle (1.6,1);
\node[scale=1, color=red] at (1.1,1.1){$\mathbf{3}$};
\node[scale=1, color=blue] at (1.5,1.1){$5$};

\draw[rounded corners] (1.7,1.2) rectangle (2.3,1);
\node[scale=1, color=red] at (1.8,1.1){$\mathbf{2}$};
\node[scale=1, color=blue] at (2.2,1.1){$4$}; 

\draw[rounded corners] (1,0.9) rectangle (1.6,0.7);
\node[scale=1, color=red] at (1.1,0.8){$\mathbf{2}$};
\node[scale=1, color=blue] at (1.5,0.8){$5$};

\draw[rounded corners] (1.7,0.9) rectangle (2.3,0.7);
\node[scale=1, color=red] at (1.8,0.8){$\mathbf{3}$};
\node[scale=1, color=blue] at (2.2,0.8){$4$};

\draw[rounded corners] (1,0.6) rectangle (1.6,0.4);
\node[scale=1, color=blue] at (1.1,0.5){$3$};
\node[scale=1, color=red] at (1.5,0.5){$\mathbf{5}$};

\draw[rounded corners] (1.7,0.6) rectangle (2.3,0.4);
\node[scale=1, color=blue] at (1.8,0.5){$2$};
\node[scale=1, color=red] at (2.2,0.5){$\mathbf{4}$}; 

\draw[rounded corners] (1,0.3) rectangle (1.6,0.1);
\node[scale=1, color=blue] at (1.1,0.2){$2$};
\node[scale=1, color=red] at (1.5,0.2){$\mathbf{5}$};

\draw[rounded corners] (1.7,0.3) rectangle (2.3,0.1);
\node[scale=1, color=blue] at (1.8,0.2){$3$};
\node[scale=1, color=red] at (2.2,0.2){$\mathbf{4}$}; 

\node[scale=1] at (1.65,-0.1){Edges in $\mathcal{H}_c$};

\end{tikzpicture}
\caption{
    An illustration of Definitions \ref{def:kikuchi} and \ref{def:HCdef}. On the left, we have a pair of hyperedges $C,C'$ which both belong to a bucket. Then, on its right is depicted an edge $S \xleftrightarrow{(C,C')} T$ created by this pair. Then, on the further right are illustrated all hyperedges of $\mathcal{H}_c$ (where $c$ is the color of the hyperedge $C$) which are created in Definition \ref{def:HCdef} by the pair $(C,C')$. Here, red and blue correspond to labels $1$ and $2$ respectively (recall that $S,T \subseteq [n] \times [2]$ and $\mathcal{H}_c$ has vertex set $[n]\times [2]$).}
    
\end{figure}

\noindent Let us now note two important properties of the $l$-Kikuchi graph $G$.
\begin{obs}\label{obs}
If $G$ contains a closed walk
whose edges are overall associated to a multi-set $\mathcal{H}'$ of hyperedges in $\mathcal{H}$, then it is such that $\bigoplus_{C \in \mathcal{H}'} C = \emptyset$. In particular, if there is some color which appears an odd number of times in $\mathcal{H}'$, then $\mathcal{H}$ contains an even cover of size at most $e(\mathcal{H}')$ in which some color appears an odd number of times. Further, $G$ has average degree at least $$e(\mathcal{H}) (m-1) \cdot \frac{\binom{k-t}{\lfloor \frac{k-t}{2} \rfloor}^2 \cdot \binom{2n-2k+2t}{l-k+t}}{\binom{2n}{l}} \geq   e(\mathcal{H}) (m-1) \cdot \left(\frac{l}{2n} \right)^{k-t}. $$
\end{obs}
\begin{proof}
We will briefly explain the average degree computation. The number of ordered pairs $(C,C')$ of hyperedges in the same bucket is $e(\mathcal{H}) (m-1)$. Now, with the pair $(C,C')$ fixed, choosing an edge $S \xleftrightarrow{(C,C')} T$ can be viewed as a two-step choice. First, we choose the intersections $\tilde{C_1} \cap S^{(1)}, \tilde{C_1} \cap T^{(1)}, \tilde{C_2} \cap S^{(2)}, \tilde{C_2} \cap T^{(2)}$ so that these intersections sizes are as required in Definition \ref{def:kikuchi} -- there are $\binom{k-t}{\lfloor \frac{k-t}{2} \rfloor}^2$ such choices. Once those intersections are chosen, the next step is to choose a set $S \cap T = S \setminus F = T \setminus F'$ of size $l-k+t$ out of $2n-2k+2t$ vertices. Combining all of these, we get that the number of edges in $G$ is $e(\mathcal{H}) (m-1) \cdot \binom{k-t}{\lfloor \frac{k-t}{2} \rfloor}^2 \cdot \binom{2n-2k+2t}{l-k+t}$. Since $G$ has $\binom{2n}{l}$ vertices, we have the desired average degree.
\end{proof}
\noindent As anticipated by the above observations, we can use Kikuchi graphs to find even covers in hypergraphs.
\begin{prop}\label{prop:kikuchiclean}
Let $k$ be odd and $\mathcal{H}$ be an edge-colored $n$-vertex $k$-uniform hypergraph with $nd$ hyperedges and an $(m,t)$-bucket decomposition. Let $G$ be the $l$-Kikuchi graph of $\mathcal{H}$ and suppose that the following holds.
\begin{itemize}
    \item For all colors $c$, hyperedges $F \in \mathcal{H}_c$ and $0 \leq j \leq k-t$, the number of hyperedges $F' \in \mathcal{H}_c$ with $|F' \cap F| = j$ is at most $ \left(d(G)/2000 \log v(G) \right)\cdot \left(n/l \right)^{k-t-j}$. 
\end{itemize}
Then, $\mathcal{H}$ contains an even cover of size at most $2\log v(G)$ 
in which some color appears an odd number of times.
\end{prop}
\begin{proof}
Our aim is to use Lemma \ref{prop:coloredcycle}. However, we might not be able to apply it directly to the Kikuchi graph $G$ (which by Observation \ref{obs} would imply an even cover in $\mathcal{H}$). This is because $G$ might not have a well-behaved coloring - it could well be that there are vertices and colors which do not satisfy the condition of Lemma \ref{prop:coloredcycle}. 

Instead, it turns out that we will be able to delete some edges from $G$ and get a subgraph $G' \subseteq G$ which does satisfy the necessary conditions. Let $d(G)$ denote the average degree of $G$ and for each ordered pair $(C,C')$ of hyperedges in $\mathcal{H}$ with colors $c,c'$, we delete all edges $S \xleftrightarrow{C,C'} T$ such that one of $S$ or $T$ is adjacent to more than $d(G)/80 \log v(G)$ edges with one of the colors $c,c'$. Let $G'$ be the resulting graph. Note that if we can show that $d(G') \geq d(G)/2$, then Lemma \ref{prop:coloredcycle} (with $s=2$) implies that $G$ contains a closed walk of size at most $2\log v(G)$ in which some color appears once. By Observation \ref{obs}, this implies an even cover in $\mathcal{H}$ of size at most $2\log v(G)$ as desired. 

The following claim clearly implies that $d(G') \geq d(G)/2$. Let $C_1,C_2$ be any two hyperedges of $\mathcal{H}$ (with colors $c_1,c_2$ respectively) in the same bucket $\mathcal{E}_i$. 
\begin{claim}
At least half of the edges with color-pair $(C_1,C_2)$ are in $G'$.
\end{claim}
\begin{proof}
To prove this claim it is enough to show that a uniformly chosen edge with color pair 
$(C_1,C_2)$ is deleted with probability at most $1/2$.
Note that with $C_1,C_2$ fixed, choosing a uniformly random edge $S \xleftrightarrow{(C_1,C_2)} T$ can be viewed as a two-step uniform choice. First, we choose at random the intersections $\tilde{C_1} \cap S^{(1)}, \tilde{C_1} \cap T^{(1)}, \tilde{C_2} \cap S^{(2)}, \tilde{C_2} \cap T^{(2)}$ so that these intersections sizes are as required in Definition \ref{def:kikuchi}. This will, for $S$ (and $T$), fix a hyperedge $F \in \mathcal{H}_{c_1} \cap \mathcal{H}_{c_2}$ which is contained in $S$ (and another hyperedge $F'$ in $\mathcal{H}_{c_1} \cap \mathcal{H}_{c_2}$ contained in $T$). Secondly, once those intersections are chosen, the next step is to choose a uniformly random set $S \cap T = S \setminus F = T \setminus F'$ of size $l-k+t$ out of $2n-2k+2t$ vertices. 

Now, we will upper bound the probability that such a random edge is not in $G'$. We will deal first with the case that $S$ is incident to 
more than $d(G)/80 \log v(G)$ other edges with the color $c_1$ (and thus, the edge $S \leftrightarrow T$ will not be in $G'$) and will show that a random edge has this property with probability at most $1/8$. The other four cases ($S$ incident to too many edges of color $c_2$, $T$ incident to too many edges of color $c_1$ and $T$ incident to too many edges of color $c_2$) are analogous and so, we will in the end have a total probability of deletion of at most $4/8 = 1/2$. Suppose then that $(C_3,C_4)$ are two other hyperedges both in the same bucket and one of them has color $c_1$ - without loss of generality, assume it is $C_3$. Suppose there exists an edge $S \xleftrightarrow{C_3,C_4} T'$ for some other $T'$. Then, $S$ must contain some additional hyperedge in $\mathcal{H}_{c_1}$ - let us denote this by $E$. For each such possible $E$, since $S \setminus F$ is chosen uniformly, the probability that $S$ contains it is at most $$\frac{\binom{2n-2k+2t}{l-|E \setminus F|-k+t}}{\binom{2n-2k+2t}{l-k+t}} \leq \left(\frac{l-k+t}{2n-2k+2t - (l-k+t)} \right)^{|E\setminus F|} \leq (0.51l/n)^{|E\setminus F|},$$ 
using the fact that in all our proofs we can assume that $l \leq n/\log n$.
Furthermore, by assumption there are at most $\left(d(G)/2000 \log v(G) \right) \cdot \left(n/l \right)^{k-t-j}$ hyperedges of $\mathcal{H}_{c_1}$ with $|E \cap F| = j$. Therefore, the expected number of such edges $S \xleftrightarrow{C_3,C_4} T'$ is at most  
$$ \left(d(G)/2000 \log v(G) \right) \cdot \sum_{1 \leq j \leq k-t} \left(n/l \right)^{k-t-j} \cdot (0.51l/n)^{k-t-j} \leq d(G)/800 \log v(G).$$
Concluding, by Markov's inequality, that the probability that a random edge is deleted because of one of its endpoints and one of its colors is at most $0.1<1/8$, as desired.
\end{proof}
\end{proof}

\section{3-LDC lower bound}

\noindent In this section, we will apply the previous tools to prove Theorem \ref{thm:LDC}.

\begin{proof}[ of Theorem \ref{thm:LDC}]
Take $K:= 10^{7}/ \alpha^2$, let $\mathcal{H}$ be such a hypergraph and let $l := n^{1/3}$. We can assume without loss of generality that we have $d := K l \log n$ colors, each with precisely $\alpha n$ hyperedges, so that our hypergraph has $\alpha nd$ hyperedges. We will want to apply Proposition \ref{prop:kikuchiclean} in order to find such a special even cover. For this, let us first note that one of the following holds.
\begin{enumerate}
    \item There is a $\mathcal{H}' \subseteq \mathcal{H}$ with $e(\mathcal{H}') \geq e(\mathcal{H})/2$ with a $(K\log n,2)$-bucket decomposition.
    \item There is a hypergraph $\mathcal{H}' \subseteq \mathcal{H}$ with $e(\mathcal{H}') \geq e(\mathcal{H})/4$ such that it has a $( \alpha d/4, 1)$-bucket decomposition and every set $R$ of size two has $\text{deg}_{\mathcal{H}'}(R) \leq K\log n$.
\end{enumerate}
\noindent 
Indeed, this follows directly by applying first Lemma \ref{lem:basiccleaning} with $t=2$ and $m = K\log n$. This gives that the first item holds or there is a $\mathcal{H}' \subseteq \mathcal{H}$ with $e(\mathcal{H}') \geq e(\mathcal{H})/2$ and such that every set $R$ of size two has $\text{deg}_{\mathcal{H}'}(R) \leq K\log n$. For this second case, we can then apply Lemma \ref{buckets} to find the $(\alpha d/4,1)$-bucket decomposition in the second item.

Now we will show that in both cases above, Proposition \ref{prop:kikuchiclean} holds for the $l$-Kikuchi graph of the corresponding hypergraph $\mathcal{H}'$, which for convenience we redenote as $\mathcal{H}$.

Let us first suppose that (1.) holds and let us consider the $l$-Kikuchi graph $G$ of $\mathcal{H}$ and the given bucket decomposition. Using that $K=10^{7}/\alpha^2$ and $v(G)={2n \choose l} \leq (2n)^l$,
by Observation \ref{obs} (with $k=3$ and $t=2$), the $l$-Kikuchi graph $G$ has average degree at least $(\alpha nd/2) \cdot (K\log n - 1) \cdot (l/2n) \geq 8000dl \log (2n) \geq 8000d \cdot \log v(G)$. Also note that for each color $c$, the multihypergraph $\mathcal{H}_c$ (from Definition \ref{def:HCdef}) is $1$-uniform, and thus consists of single vertices in $[n] \times [2]$. Therefore, by Proposition \ref{prop:kikuchiclean} we need only to show that given a color $c$, an $F \in \mathcal{H}_c$ and a $0 \leq j \leq 1$, the number of $F' \in \mathcal{H}_c$ with $|F' \cap F| = j$ is at most 
$4d \cdot \left(n/l \right)^{1-j} \leq d(G)/(2000 \log v(G)) \cdot \left(n/l \right)^{1-j}$.

For $j = 0$, the number of such $F'$ is at most four times the number of hyperedges in $\mathcal{H}_c$ which is $4(K\log n) n = 4d \cdot (n/l)$. To see this, recall that the hypergraph $\cal H$ is properly colored and therefore at each of its vertices there is only one bucket of size $K\log n$ containing a hyperedge $C$ of color $c$ - every other hyperedge $C'$ in this bucket together with $C$ gives then $4$ hyperedges of $\mathcal{H}_c$ - indeed there are four hyperedges of $\mathcal{H}_c$ satisfying one of the items in Definition \ref{def:HCdef}; these consist of the two labeled copies of the vertex in $C \setminus C'$ and of the vertex in $C' \setminus C$. For $j=1$, let the hyperedge $C \in E({\mathcal H})$ of color $c$ be the one in Definition \ref{def:HCdef} forming $F$. The hyperedge $F' \in \mathcal{H}_c$ with $F' = F$  can be obtained in two ways: first from some hyperedge $C'$ of $\cal H$ incident to the same vertex as singleton the (uncolored) $F$; this gives at most the maximum degree of $\cal H$, i.e., $d$ hyperedges; we can also obtain $F'$ from any hyperedge $C' \in {\mathcal H}$ from the same bucket as $C$. This gives at most $K\log n$ additional hyperedges. So in total $d + K \log n \leq 2d$.  
 
Now assume that (2.) holds and recall that $l=n^{1/3}$,  $d=Kl \log n$ and $v(G)\leq (2n)^l$ where $G$ is the $l$-Kikuchi graph of $\mathcal{H}$ and the given bucket decomposition.
In this case ${\mathcal H}_c$ is a multigraph for every color $c$.
Observation \ref{obs} (with $k=3$ and $t=1$) implies that the $l$-Kikuchi graph has average degree at least 
\begin{align*}
d(G) &\geq (\alpha nd/4) \cdot (\alpha d/4 - 1) \cdot (l/2n)^2\\
&\geq (\alpha K)^2/100 \cdot  (l^3/n) \cdot l \cdot (\log n)^{2}\\
&\geq 10^5 K \cdot \log n \cdot \log v(G).
\end{align*}
Hence $d(G)/2000\log v(G) \geq 50K \log n$.  Take $F=({\textcolor{red}{u}},{\textcolor{blue}{v}}) \in {\mathcal{H}_c}$ which corresponds (as in Definition \ref{def:HCdef}) to a pair of hyperedges $C,C'$ of $\cal H$ from some bucket such that $C$ has color $c$ and w.l.o.g contains $u$ and $C'$ contains $v$. We need to verify that for  $0 \leq j \leq 2$ 
the number of $F' \in \mathcal{H}_c$ with $|F' \cap F| = j$ is at most $50K\log n (n/l)^{2-j}$, i.e., satisfies the condition in Proposition \ref{prop:kikuchiclean} (with $k=3$ and $t=1$). For $j = 0$, this number is the number of hyperedges in $\mathcal{H}_c$, which  is at most $8n$ times the size of each bucket, that is, $8n \cdot (\alpha d/4)=2\alpha K nl\log n \leq 2K \log n (n/l)^2$. This is because each bucket has $\alpha d/4$ and so if it contains a hyperedge of color $c$ then it produces $\alpha d/4$ pairs of hyperedges with one of color $c$. Each such pair then contributes $8$ hyperedges to ${\mathcal H}_c$ (see Figure \ref{fig:gadget}). Furthermore, since the coloring of $\cal H$ is proper there are at most $n$ buckets containing a hyperedge of color $c$.

For $j = 1$, note that the number of $F'$ is at most $16$ times the maximum degree of $\cal H$ plus $16$ times the size of each bucket  which is at most $16d + 16\alpha d/4 \leq n /l$. To see this note that we can get $F'$ in two ways; one is by taking the unique hyperedge in $\cal H$ of color $c$ through either $u$ or $v$ and combining it with any other hyperedge from the same bucket (giving a pair of hyperedges for Definition \ref{def:HCdef}); the other is by considering any hyperedge of $\cal H$ containing either $u$ or $v$ and combining it with a hyperedge of color $c$ in the same bucket (giving a pair of hyperedges in $\cal H$ for Definition \ref{def:HCdef}); like before, in each case, each pair produces $8$ hyperedges of $\mathcal{H}$. Finally, for $j=2$, note that $F' \in \mathcal{H}_c$ with $F' = F$ can be obtained by taking a unique hyperedge $C$ of color $c$ through $u$ (or $v$) and combining it with any hyperedge in the same bucket which contains $v$. This hyperedge must intersect the vertex of $C-\{u\}$ which is incident to every other hyperedge in the bucket containing $C$ and so, since the codegree of every pair of vertices is at most $K\log n$ this gives at most $2K\log n$ options.
\end{proof}

\section{Even covers in hypergraphs for odd \texorpdfstring{$k$}{k}}

\noindent In this section, we prove Theorem \ref{thm:evencover} (ii).
In this section, we will always consider odd uniformity $k$.
The $(\log n)^{\frac{1}{k+1}}$ factor improvement (as opposed to $\log n$) follows from a \emph{new} variant of the Kikuchi graph from Definition \ref{def:kikuchi}, which we describe next.

\subsection{The Flower Kikuchi graph and even covers}\label{sec:flowerkikuchi}
\noindent We now introduce the \emph{flower Kikuchi graph}. Let $\mathcal{H}$ be an $n$-vertex $k$-uniform hypergraph with $nd$ hyperedges and minimum degree $\delta$ and such that every set of size $k-1$ is contained in at most one hyperedge. Suppose that the hyperedges of $\mathcal{H}$ are colored in \emph{red} and \emph{blue}. Suppose also that for each vertex $v$ in $\mathcal{H}$, we fix a collection $\mathcal{E}_v$ of precisely $\delta$ hyperedges containing $v$.


\begin{defn}
A \emph{flower gadget} $F = (C, P_1,\dots, P_k)$ in $\mathcal{H}$ is a collection of $k+1$ hyperedges in $\mathcal{H}$ with a \emph{center} hyperedge $C = \{v_1,\ldots, v_k\}$ and $k$ \emph{petal} hyperedges $P_1, \ldots, P_k$ which are disjoint and such that each $P_i$ is contained in $\mathcal{E}_{v_i}$ with $P_i \cap C = \{v_i\}$. $F$ is a \emph{good flower gadget} if $C$ is \emph{blue} and all hyperedges $P_i$ are \emph{red}. See Figure \ref{fig:gadget} for an illustration.
\end{defn}

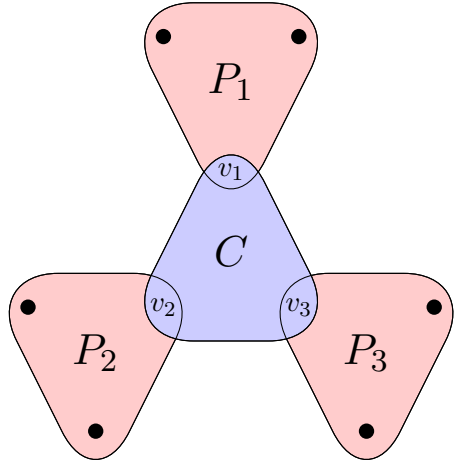
\begin{figure}[ht]
\begin{tikzpicture}[scale = 3]

    \draw [rounded corners=10mm,fill=red!20] (0,1.5)--(0.5,0.5)--(1,1.5) -- cycle;
    \draw [rounded corners=10mm,fill=red!20] (1.1,-0.7)--(0.6,0.3)--(1.6,0.3) -- cycle;
    \draw [rounded corners=10mm,fill=red!20] (-0.1,-0.7)--(-0.6,0.3)--(0.4,0.3) -- cycle;

    \draw [rounded corners=10mm,fill=blue!20] (0,0)--(0.5,1)--(1,0) -- cycle;

    \draw [rounded corners=10mm] (-0.1,-0.7)--(-0.6,0.3)--(0.4,0.3) -- cycle;
    \draw [rounded corners=10mm] (0,0)--(0.5,1)--(1,0) -- cycle;
    \draw [rounded corners=10mm] (1.1,-0.7)--(0.6,0.3)--(1.6,0.3) -- cycle;
    \draw [rounded corners=10mm] (0,1.5)--(0.5,0.5)--(1,1.5) -- cycle;
    
    \node[scale=1, color=black] at (0.2,0.15){$v_2$};
    \node[scale=5, color=black] at (0.2,1.35){$.$};
    \node[scale=1, color=black] at (0.8,0.15){$v_3$};
    \node[scale=5, color=black] at (0.8,1.35){$.$};
    \node[scale=1, color=black] at (0.5,0.75){$v_1$};
    \node[scale=5, color=black] at (-0.4,0.15){$.$};
    \node[scale=5, color=black] at (1.4,0.15){$.$};
    \node[scale=5, color=black] at (1.1,-0.4){$.$};
    \node[scale=5, color=black] at (-0.1,-0.4){$.$};

    \node[scale=1.5, color=black] at (0.5,0.4){$C$};
    \node[scale=1.5, color=black] at (0.5,1.15){$P_1$};
    \node[scale=1.5, color=black] at (-0.1,-0.05){$P_2$};
    \node[scale=1.5, color=black] at (1.1,-0.05){$P_3$};
\end{tikzpicture}
\caption{
    An illustration of a good flower gadget with $k=3$.}
    
    \label{fig:gadget}
\end{figure}

\begin{defn}
The \emph{flower $l$-Kikuchi graph} $G$ of $\mathcal{H}$ has vertex set $V:= \binom{[n]}{l}$ and its edges are defined and colored as follows. For each good flower gadget $F = (C, P_1,\dots, P_k)$, define an edge between $S, T \in V$, denoted by $S \xleftrightarrow{C} T$ and colored with $C$, if $S \oplus T$ is equal to $P_1\oplus\dots\oplus P_k \oplus C$, and 
$|S\cap (P_i \setminus v_i)| = |T\cap (P_i \setminus v_i)|=\frac{k-1}{2}$ for all $i$.
\end{defn}
\noindent We first prove the following simple claim.
\begin{claim}
An edge $S \leftrightarrow T$ can possibly be created by at most $2^{k^3}$ good flower gadgets. 
\end{claim}
\begin{proof}
To prove this, we only need to use the property that  no set of size $k-1$ is contained in more than one hyperedge of the hypergraph $\mathcal{H}$. Observe that for a good flower gadget $\{C,P_1,\ldots,P_k\}$ to create this edge, it must be that $|(S \setminus T) \cap A| = |(T \setminus S) \cap A|  = \frac{k-1}{2}$ for all $A \in \{P_1 \setminus \{v_1\},\ldots,P_k \setminus \{v_k\}\}$. Since $|S \setminus T| = |T \setminus S| = \frac{k(k-1)}{2}$, there are at most $ \binom{k(k-1)/2}{(k-1)/2}^{2k} \leq 2^{k^3}$ ways to define these sets $(S \setminus T) \cap A, (T \setminus S) \cap A$, giving us the sets $P_1 \setminus {v_1}, \ldots , P_k \setminus {v_k}$. Since the hypergraph is such that no set of size $k-1$ is contained in more than one edge, through each of these sets there is a unique edge, giving us $P_1,\ldots,P_k$ (up to permutation of names which does not matter). Then the vertex set 
$(P_1 \cup \ldots \cup P_k) \setminus (S\cup T)$ gives $C$. 
\end{proof}
Given the above claim, we fix the coloring arbitrarily by coloring $S \xleftrightarrow{C} T$ with only one of the at most $2^{k^3}$ possible options. We make note of some other crucial (but simple to check) properties of $G$ just as we did previously in Section \ref{sec:kikuchi} with Observation \ref{obs}. 
\begin{obs}\label{obs:flower}
If $G$ contains a closed walk
and we denote by $\{C^{(j)},P^{(j)}_1, \ldots, P^{(j)}_k\}$ the good flower gadgets which create the edges of the walk, then $\bigoplus_{j} \left( P^{(j)}_1\oplus\dots\oplus P^{(j)}_k \oplus C^{(j)} \right) = \emptyset$. In particular, if there is some (blue) edge which appears an odd number of times among the $C^{(j)}$, then $\mathcal{H}$ contains an even cover whose size is of the same order as the length of the walk. Also, $G$ has $N := \binom{n}{l}$ vertices and letting $m$ denote the number of good flower gadgets in $\mathcal{H}$, then $G$ has average degree
$$\frac{m \cdot \binom{n-k(k-1)}{l - k(k-1)/2}}{2^{k^3}N} \geq \frac{m}{10^{k^3}} \cdot \left(\frac{l}{n} \right)^{k(k-1)/2} .$$ 
\end{obs}
\begin{proof}
We will briefly explain the average degree computation. Each good flower gadget creates at least $\binom{n-k(k-1)}{l - k(k-1)/2}$ edges since after choosing the intersections $S \cap (P_i \setminus \{v_i\}), T \cap (P_i \setminus \{v_i\})$ we are left choosing $S \cap T$ which can be any set of size $l - \frac{k(k-1)}{2}$ from $V(\mathcal{H}) \setminus \bigcup_i (P_i \setminus \{v_i\})$, which has size $n-k(k-1)$. Further, as discussed above, each edge is possibly created by at most $2^{k^3}$ good gadgets, and therefore, $G$ has at least $\frac{m \cdot \binom{n-k(k-1)}{l - k(k-1)/2}}{2^{k^3}}$ edges. The last estimate can be verified by considering separately cases when $l \leq k^2$ or $l>k^2$ and recalling that we can assume that $l \leq n/\log n$.
\end{proof}
\noindent We will show an analog of Proposition \ref{prop:kikuchiclean}, that if the flower $l$-Kikuchi graph has some 'nice' properties, then the underlying hypergraph contains a small even cover. Since we will not need the full power of such a statement, in the next section, we will only prove a simpler version of it which we use to establish Theorem \ref{thm:evencover} (ii).

\subsection{Proof of Theorem \ref{thm:evencover} (ii)}

Let $\mathcal{H}$ be a $k$-uniform $n$-vertex hypergraph with $nd$ hyperedges where $d \geq C \left(n/l \right)^{k/2-1} (\log n)^{\frac{1}{k+1}}$ for some sufficiently large constant $C$ which depends on $k$. We take $n$ large enough so that any fixed power of $\log n$ which appears in the proof is bigger than any constant dependent on $k$ appearing. Note first that using Corollary \ref{cor} we can assume that $\mathcal{H}$ is such that every set of more than $\frac{k-1}{2}$ vertices (since $k$ is odd) has co-degree at most $1$. Without loss of generality let us now assume that $d = C \left(n/l \right)^{k/2-1} (\log n)^{\frac{1}{k+1}}$. We apply Lemma \ref{prop:cleaning1} with $m$ defined as $m(1) = d/10k$ and $m(t) = (d/10k)^{\frac{k/2-t}{k/2-1}} \cdot (\log n)^{1-\frac{1}{k+1}}$ for $2 \leq t \leq \frac{k-1}{2}$. This gives us a sub-hypergraph $\mathcal{H}_0$ with at least $nd/2k$ hyperedges and a $t \leq \frac{k-1}{2}$ with the given properties in the lemma. Let us split the remainder of the proof into the cases $t > 1$ and $t=1$.

\vspace{0.2cm}
\noindent \textbf{Case 1: $t > 1$}
\vspace{0.2cm}

\noindent In this case, we know that $\mathcal{H}_0$ has an $(m(t),t)$-bucket decomposition and $\text{deg}_{\mathcal{H}_0}(S) < m(|S|)$ for all sets $S$ with $|S| > t$. We will show that this hypergraph satisfies the conditions of Proposition \ref{prop:kikuchiclean}, which will immediately give an even cover as desired. Let then $G$ be the $l$-Kikuchi graph for $\mathcal{H}_0$ and let $m := m(t)$. By Observation \ref{obs}, it has average degree $d(G) \geq e(\mathcal{H}_0)(m-1) \left(l/2n \right)^{k-t}$. Recalling that $G$ has $v(G) \leq \binom{2n}{l} \leq (2n)^l$ vertices, we have that 
\begin{eqnarray}
\label{degree}
    d(G) &\geq& \frac{nd}{2k} \cdot m  \cdot (l/2n)^{k-t}\\
    &\geq &  \frac{nd}{2k} \cdot (d/10k)^{\frac{k/2-t}{k/2-1}} \cdot (\log n)^{1-\frac{1}{k+1}} \cdot (l/2n)^{k-t}  \\
&\geq&  (\log n)^{1/k^2} \cdot (l \log (2n)) \geq (\log n)^{\frac{1}{(k+1)^2}}  \cdot \log v(G). \nonumber
\end{eqnarray}
Now, let us redefine $\mathcal{H}$ as $\mathcal{H}_0$. View $\mathcal{H}$ as an edge-colored hypergraph by coloring each hyperedge with a different color and let us verify the conditions of Proposition \ref{prop:kikuchiclean}. For each hyperedge $C \in \mathcal{H}$, the hypergraph $\mathcal{H}_C$ is then formed by taking every hyperedge $C'$ in the same bucket as $C$ and putting all hyperedges $F$ in $\mathcal{H}_C$ which satisfy one of the items in Definition \ref{def:HCdef} -- in particular, part of $F$ (either $\lceil \frac{k-t}{2} \rceil$ or $\lfloor \frac{k-t}{2} \rfloor$) will be contained in either $C^{(1)}$ or $C^{(2)}$ and the rest will be contained in either $C'^{(1)}$ or $C'^{(2)}$. Depending on the parity of $k-t$, each such $C'$ contributes  to at least $2 \binom{k-t}{\lfloor \frac{k-t}{2} \rfloor }^2$ and at most $4 \binom{k-t}{\lfloor \frac{k-t}{2} \rfloor }^2\leq 2^{2k}$ hyperedges for $\mathcal{H}_C$.

Let us fix a hyperedge $C$ in $\mathcal{H}$, $F \in \mathcal{H}_C$ and $0 \leq j \leq k-t$. Let us denote the bucket of $\mathcal{H}$ containing $C$ by $(X_i,\mathcal{E}_i)$. We bound the number of $F' \in \mathcal{H}_C$ with $|F' \cap F| = j$ as follows. Since $F \in \mathcal{H}_C$, there is some hyperedge $C'$ in the same bucket as $C$ such that one of the items in Definition \ref{def:HCdef} applies. Let us assume without loss of generality that $\{|\tilde{C}^{(1)} \cap F^{(1)}|, |\tilde{C}'^{(2)} \cap F^{(2)}|\} = \{\ceil{\frac{k-t}{2}}, \floor{\frac{k-t}{2}}\}$. Fix a pair $j_1,j_2$ with $j = j_1 + j_2$ and let us consider those $F'$ with $j_1 = |F'^{(1)} \cap F^{(1)}| \leq \lceil \frac{k-t}{2} \rceil$ and $j_2 = |F'^{(2)} \cap F^{(2)}| \leq \lceil \frac{k-t}{2} \rceil$. We first fix for $F'$ one of the items in Definition \ref{def:HCdef} which applies. Let us first count those with $\{|\tilde{C}^{(1)} \cap F'^{(1)}|, |\tilde{C}''^{(2)} \cap F'^{(2)}|\} = \{\ceil{\frac{k-t}{2}}, \floor{\frac{k-t}{2}}\}$ for some hyperedge $C''$ in the same bucket as $C$ and $C'$. 

First, we have at most $2^k$ options for the choice of $F'^{(1)}$ since it must be a subset of $C^{(1)}$. For the choice of $F'^{(2)}$, note that after choosing an appropriate hyperedge $C''$, we have again at most $2^k$ options for the choice of $\tilde{C}''^{(2)} \cap F'^{(2)}$. But now the choice of $C''$ is restricted to the fact that $X_i \cup (F'^{(2)} \cap F^{(2)})$ (which is here considered naturally as a set of vertices in $[n]$ and not in $[n] \times [2]$) must be contained in $C''$ and thus, since $|X_i \cup (F'^{(2)} \cap F^{(2)})| = t+j_2$, the conditions of Lemma \ref{prop:cleaning1} imply that there are at most $$2^k \cdot \max_{|R| = t+j_2} \text{deg}_{\mathcal{H}'}(R) \leq 2^k \cdot m(t+j_2) \leq 2^k \cdot   (d/10k)^{\frac{k/2-t-j_2}{k/2-1}} \cdot (\log n)^{1-\frac{1}{k+1}}$$ choices for such an hyperedge $C''$, where the extra $2^k$ factor serves as a gross upper bound for the number of possible subsets $X_i \cup (F'^{(2)} \cap F^{(2)}) \subseteq F$. Combining all these consideration, we have at most 

\begin{eqnarray*}
2^{3k} \cdot   (d/10k)^{\frac{k/2-t-j_2}{k/2-1}} \cdot (\log n)^{1-\frac{1}{k+1}} &\leq&  \log n \cdot (n/l)^{k/2-t-j_2} \\
&\leq& \log n \cdot (n/l)^{k-t-j-\frac{1}{2}}
\end{eqnarray*}
options for such $F'$. Notice that we used in the previous inequality that $k/2 - t - j_2 \leq (k/2+j_1)-t-j \leq (k/2 + \lceil \frac{k-t}{2} \rceil) -t - j\leq (k/2 + \lceil \frac{k-2}{2} \rceil) -t - j \leq (k/2 + \frac{k-1}{2}) -t - j \leq  k-t-j - \frac{1}{2}$, since $t \geq 2$. 

Now, the arguments and the upper bound are analogous in the other cases for $F'$ (that is, for all possible values of $j_1,j_2$ and whether $\{|\tilde{C} \cap F'^{(1)}|, |\tilde{C}'' \cap F'^{(2)}|\} = \{\ceil{\frac{k-t}{2}}, \floor{\frac{k-t}{2}}\}$ or $\{|\tilde{C} \cap F'^{(2)}|, |\tilde{C}'' \cap F'^{(1)}|\} = \{\ceil{\frac{k-t}{2}}, \floor{\frac{k-t}{2}}\}$). Since there are at most $2k^2$ such cases, by using (\ref{degree}) and choosing $C$ large enough as a function of $k$, we conclude that the number of choices for $F'$ is in total at most $$2k^2 \cdot \log n \cdot (n/l)^{k-t-j-\frac{1}{2}}   \leq (d(G)/2000\log v(G)) \cdot (n/l)^{k-t-j}, $$ as desired, where we used that $\log n \leq \sqrt{n/l}$. Hence, the property of Proposition \ref{prop:kikuchiclean} is satisfied and therefore $\mathcal{H}$ contains an even cover of size at most $2\log v(G) = O(l \log n)$, as desired.

\vspace{0.2cm}
\noindent \textbf{Case 2: $t = 1$}
\vspace{0.2cm}

\noindent In this case we know that $\mathcal{H}_0$ has minimum degree at least $d/10k$ and $\text{deg}_{\mathcal{H}_0}(S) < m(|S|) = (d/10k)^{\frac{k/2-|S|}{k/2-1}} \cdot (\log n)^{1-\frac{1}{k+1}} \leq (n/l)^{k/2-|S|} \cdot (\log n)^{1-1/k^2}$ for all sets $S$ with $|S| > 1$ and $m(1) = d/10k$. As before, let us redefine $\mathcal{H}$ as $\mathcal{H}_0$. For each vertex $v$, let $\mathcal{E}_v$ denote a collection of precisely $d/10k$ hyperedges of $\mathcal{H}$ incident on $v$. Now, let us independently and uniformly at random color each hyperedge of $\mathcal{H}$ as either red or blue. We then have the following.
\begin{claim}
With positive probability, there are at least $nd^{k+1}/(100k)^k$ good flower gadgets.
\end{claim}
\begin{proof}
We show that the number of flower gadgets in $\mathcal{H}$ is at least $nd^{k+1}/(20k)^k$. Then, since each flower gadget is also good with probability $1/2^{k+1}$ then the expected number of good flower gadgets is at least $nd^{k+1}/(100k)^k$ (using $5^k > 2^{k+1}$) and further, with positive probability this holds. 

Now, let us count the number of flower gadgets $\{C,P_1, \ldots, P_k\}$. We have $e(\mathcal{H}) \geq nd/2k$ options for the hyperedge $C = \{v_1, \ldots, v_k\}$. Then, since $P_1 \in \mathcal{E}_{v_1}$ and it is disjoint to $C \setminus \{v_1\}$, there are $|\mathcal{E}_{v_1}| - |C| \cdot m(2) > d/10k - km(2)$ options for the edge $P_1$. Subsequently, $P_2$ must be an edge belonging to $\mathcal{E}_{v_2}$ which is disjoint to $(C \cup P_1) \setminus \{v_2\}$ and therefore, we have at least $|\mathcal{E}_{v_2}| - |C \cup P_1| \cdot m(2) > d/10k - 2km(2)$ options for $P_2$. Continuing in this manner, we have that the number of flower gadgets is at least $$(nd/2k) \prod_{1 \leq i \leq k} (d/10k-ikm(2)) .$$ 
Since $m(2) = O_k(1) \cdot  (n/l)^{k/2-2} \cdot (\log n) \leq o(1) \cdot d/10k^3$, this product is at least $nd^{k+1}/(20k)^k$, as desired.
\end{proof}
\noindent Now, consider the flower $l$-Kikuchi graph $G$ of $\mathcal{H}$. By the previous claim and Observation \ref{obs:flower}, it has average degree at least $\frac{nd^{k+1}/(100k)^k}{10^{k^3}} \cdot \left(\frac{l}{n} \right)^{k(k-1)/2} \geq  \frac{nd^{k+1}}{100^{k^3}} \cdot \left(\frac{l}{n} \right)^{k(k-1)/2} \geq \frac{C^{k+1}}{100^{k^3}} \cdot l \log n \geq C \log N$. Our aim is now to apply Lemma \ref{prop:coloredcycle} to the graph $G$. Indeed, by Observation \ref{obs:flower}, a closed walk in $G$ such that some color appears only once implies an even cover in $\mathcal{H}$ of the same order. However, it might be that the graph $G$ does not satisfy the conditions of Lemma \ref{prop:coloredcycle} and so, we need to first do a cleaning procedure similar to what was done in the proof of Proposition \ref{prop:kikuchiclean}. Indeed, we delete all edges $S \xleftrightarrow{C} T$ of $G$ such that one of $S$ or $T$ are incident to another edge of color $C$. Let $G'$ denote the resulting graph and note the following.
\begin{lem} \label{lem:deletion-process}
$G'$ has average degree at least $(C/2) \log N$.
\end{lem}
\begin{proof}
Fix a good flower gadget $F = (C, P_1,\dots, P_k)$ and subsets $L_1,\dots,L_k$ and $R_1,\dots,R_k$ such that each $|R_i \cap (P_i \setminus \{v_i\})|$ and $|L_i\cap (P_i \setminus \{v_i\})|$ are equal to $\frac{k-1}{2}$. Next, we choose a uniformly random edge $S \xleftrightarrow{C} T$ of $G$ such that $S\cap (P_i \setminus \{v_i\}) = L_i$ and $T\cap (P_i \setminus \{v_i\}) = R_i$ for all $i\in[k]$. As in the definition, we use $v_i$ to denote the vertex at which $P_i$ intersects $C$. We will show that the probability that $S \xleftrightarrow{C} T$ is deleted is at most $1/2$, which gives the desired outcome.

In order for $S$ to be incident to another edge of color $C$, there must be $i\in[k]$ and a hyperedge $P \ne P_i$ in $\mathcal{E}_{v_i}$ such that $|(P\setminus \{v_i\})\cap S| = \frac{k-1}{2}$. For each such $P$, since the intersection $S \cap (P_i \setminus \{v_i\}) = L_i$ is already chosen, and further, the set $S' := S \setminus ((P_1 \setminus \{v_1\}) \cup \ldots (P_k \setminus \{v_k\}))$ is a uniformly random set of size $l - \frac{k(k-1)}{2}$ in $V(\mathcal{H}) \setminus ((P_1 \setminus \{v_1\}) \cup \ldots (P_k \setminus \{v_k\}))$. Therefore, the probability that $|(P\setminus \{v_i\}) \cap S| = \frac{k-1}{2}$ is equal to the probability that the set $S'$ intersects $P \setminus P_i$ in $\frac{k-1}{2} - |(P\setminus \{v_i\}) \cap (P_i\setminus \{v_i\})| = \frac{k-1}{2} - |P \cap P_i| + 1$ vertices. Since $S'$ is chosen uniformly, this probability is at most 
$$2^k \cdot \frac{\binom{n-(k-1)(k+\frac{1}{2})+ |P \cap P_i|-1}{l-\frac{(k+1)(k-1)}{2} + |P \cap P_i| - 1}}{\binom{n-k(k-1)}{l-\frac{k(k-1)}{2}}} \leq 2^k \cdot (l/10n)^{(k-1)/2 - |P \cap P_i| + 1},$$
where the $2^k$ is an upper bound for the number of intersections between $S'$ and $P \setminus P_i$. Now we count for each $j \geq 1$, the number of such $P$ with $|P \cap P_i| = j$. For $j = 1$, this is clearly at most $|\mathcal{E}_{v_i}| = d/10k = m(1) \leq (n/l)^{k/2-1} \cdot (\log n)^{1-1/k^2}$. Moreover, for $j \geq 2$, this is at most $\max_{|X| = j} \text{deg}_{\mathcal{H}} (X) \leq m(j) \leq (n/l)^{k/2-j} \cdot (\log n)^{1-1/k^2}$. Further, by the assumption that no two hyperedges in $\mathcal{H}$ intersect in more than $k/2$ vertices, this number is $0$ for $j \geq \frac{k+1}{2}$. Hence, combining these observations, the probability of $S$ having other $C$-colored incident edges is at most
\begin{align*}
    &2^k \sum_{1 \leq j \leq \frac{k-1}{2}} m(j) \cdot (l/10n)^{(k-1)/2 - j + 1} \\
    &\leq \sum_{1 \leq j \leq \frac{k-1}{2}} (n/l)^{k/2-j} \cdot (\log n)^{1-1/k^2}  \cdot (l/n)^{(k-1)/2 - j + 1} \\
    &\leq k \sqrt{\frac{l}{n}} \cdot (\log n)^{1-1/k^2}  \leq  1/4 ,
\end{align*}

where we used that $l \leq n/ (\log n)^2 $ and that $n$ is sufficiently large. Applying the same argument to $T$ gives the desired outcome.
\end{proof}
\noindent To finish, note that the edge-coloring in $G'$ is now a proper edge-coloring and further, it has average degree at least $20 \log N$ and thus, Lemma \ref{prop:coloredcycle} implies that $G$ contains a closed walk of size $O(\log N)$ such that some color appears exactly once. By Observation \ref{obs:flower}, this implies an even cover in $\mathcal{H}$ of size $O(l \log n)$, as desired.
\qed

\section{Concluding Remarks}
In this work, we gave a simple proof without the use of matrix concentration inequalities that recovers and improves on two recent developments on finding even covers in edge-colored hypergraphs. To conclude, we would like to point out two natural directions for further progress: 

\begin{itemize}
    \item \textbf{Even Covers in Hypergraphs}: Theorem~\ref{thm:evencover} is short of Feige's original conjecture~\cite{Feige08} that asks to show that every $k$-uniform hypergraph with $m \geq C n (n/\ell)^{k/2-1}$ hyperedges for some absolute constant $C>0$ contains an even cover of size $O(\ell \log n)$. We believe that a strengthening of our techniques via some appropriate generalization of flower Kikuchi graphs might be enough to resolve this remaining slack. On the flip side, finding explicit constructions of hypergraphs with $m = C n (n/\ell)^{k/2-1}$ hyperedges that avoid $c \ell \log n$ length even covers for some constant $c>0$ remains a natural open question. 

    \item \textbf{Lower bounds on Locally Decodable Codes:} Theorem~\ref{thm:LDC} shows an almost cubic lower bound on binary linear codes with constant distance and three query local decodability. This is still significantly far from the best known construction of three query LDCs due to Yekhanin and Efremenko that achieves a sub-exponential length. Finding methods that could prove a super-polynomial (or at least super-cubic) lower bounds is an outstanding open question. We note that in a recent work (also based on Kikuchi matrices) an exponential lower bound~\cite{kothari2023exponential} (improving on the prior best cubic lower bounds above) was in fact obtained for the stronger setting of three query linear \emph{Locally Correctable Codes} -- a stricter variant of locally decodable codes where every bit of the codeword can be decoded by reading only $3$-bits of the received corrupted codeword. Their methods, however, strongly exploit the additional structure offered by the local correction property. 
\end{itemize}


\end{document}